\newif\ifarxiv
\arxivtrue

\ifarxiv
   \newcommand\IfApx[2]{#1}
\else
   \newcommand\IfApx[2]{#2}
\fi

\newif\ifdraft     %
\draftfalse

\newif\ifexternal  %
\externaltrue

\documentclass{article}
\title{How Hard is it to Decide if a Fact is Relevant to a Query?}
\author{%
Meghyn Bienvenu
\and
Diego Figueira
\and
Pierre Lafourcade
\\
\affiliations
Univ. Bordeaux, CNRS,  Bordeaux INP, LaBRI, UMR 5800, F-33400, Talence, France\\
\emails{
\{meghyn.bienvenu, diego.figueira, pierre.lafourcade\}@u-bordeaux.fr}
}

\usepackage{kr}

\usepackage{times}
\usepackage{soul}
\usepackage{url}
\usepackage[hidelinks]{hyperref}
\usepackage[utf8]{inputenc}
\usepackage[small]{caption}
\usepackage{graphicx}
\usepackage{amsmath}
\usepackage{amsthm}
\usepackage{booktabs}
\usepackage{algorithm}
\usepackage{algorithmic}
\ifarxiv%
\usepackage[bibliography=common,appendix=append]{apxproof} %
\else
\usepackage[bibliography=common,appendix=strip]{apxproof}
\fi

\newif\ifappendix  %
\appendixfalse

\usepackage{scrhack} %
\usepackage{listings} %
\usepackage{tabularx} %
\usepackage{multirow} %
\usepackage{subcaption}

\usepackage{graphicx} %
\usepackage[table,xcdraw]{xcolor}

\usepackage{tikz}
\usepackage{tikzit} %
\usetikzlibrary{calc}
\usetikzlibrary{decorations.pathreplacing,calligraphy}
\usetikzlibrary{decorations.pathmorphing}
\usetikzlibrary{positioning}
\usetikzlibrary{patterns.meta}
\usetikzlibrary {arrows.meta}

\usepackage{bbding} %
\usepackage{csquotes} %
\usepackage[inline,shortlabels]{enumitem} %
\usepackage{lineno} %
\usepackage{xspace} %

\usepackage{hyperref}
\usepackage{cleveref} %
\usepackage{crossreftools} %

\usepackage{amssymb} %
\usepackage{amsthm} %
\usepackage{centernot} %
\usepackage{mathcommand} %
\usepackage{mathtools} %
\usepackage{nicefrac} %
\usepackage{stmaryrd} %
\ifdraft %
\usepackage[backgroundcolor=orange!20, textcolor={Dark Ruby Red}, textsize=tiny]{todonotes}
\else
\usepackage[backgroundcolor=orange!20, textcolor={Dark Ruby Red}, textsize=tiny,disable]{todonotes}
\fi
\definecolor{green}{RGB}{0,120,0}
\definecolor{hlyellow}{RGB}{250, 250, 190}
\definecolor{diegoeditcolor}{RGB}{210,210,255}
\definecolor{meghyneditcolor}{RGB}{210,255,210}
\definecolor{pierreeditcolor}{RGB}{255, 225, 186}
\newcommand{\sidediego}[1]{}
\newcommand{\sidemeghyn}[1]{}
\newcommand{\sidepierre}[1]{}
\newcommand{\meghyn}[1]{}
\newcommand{\pierre}[1]{}

\newcommand{\diego}[1]{}

\newcommand{\aka}{a.k.a.}

\newcommand{\ie}{i.e.,}

\providecommand{\st}{}\renewcommand{\st}{s.t.} %

\DeclarePairedDelimiter\set{\{}{\}} %

\LoopCommands\lettersUppercase[l#1] %
   {\newmathcommand#2{\mathbb{#1}}}

\LoopCommands{ABCDEFGIJKMNQRTUVWXYZ}[#1] %
   {\newmathcommand#2{\mathcal{#1}}}
\LoopCommands{HLOPS}[#1] %
   {\renewmathcommand#2{\mathcal{#1}}}

\newcommand{\defeq}{\vcentcolon=}

\renewcommand{\le}{\leqslant}
\renewcommand{\ge}{\geqslant}

\newcommand{\inc}{\subseteq} %

\newcommand{\ic}{\sqsubseteq} %

\newrobustcmd\decisionproblem[3]{
	\begin{center}\AP
	\fbox{\begin{tabular}{rl}
	\multicolumn{2}{l}{\parbox[t]{.9\linewidth}{#1}} \\
    \hline
	{\emph{Input}}: & \parbox[t]{.7\linewidth}{#2} \\   
	{\emph{Question}}: & \parbox[t]{.7\linewidth}{#3}
	\end{tabular}} 
	\end{center}
}

\definecolor{light-gray}{gray}{0.9} %
\newcommand{\proofcase}[1]{\noindent\setlength{\fboxsep}{2pt}\colorbox{light-gray}{#1\hspace{.2em}}~} %

\newenvironment{equation-inline}{ %
\refstepcounter{equation}}{\hfill(\theequation)\\}

\theoremstyle{plain}

\newtheoremrep{theorem}{Theorem}[section]

\newtheoremrep{claim}[theorem]{Claim}
   \AddToHook{env/claim/begin}{\crefalias{theorem}{claim}} 
\newtheoremrep{corollary}[theorem]{Corollary}
   \AddToHook{env/corollary/begin}{\crefalias{theorem}{corollary}} 
\newtheoremrep{example}[theorem]{Example}
   \AddToHook{env/example/begin}{\crefalias{theorem}{example}} 
\newtheoremrep{lemma}[theorem]{Lemma}
   \AddToHook{env/lemma/begin}{\crefalias{theorem}{lemma}} 
\newtheoremrep{proposition}[theorem]{Proposition}
   \AddToHook{env/proposition/begin}{\crefalias{theorem}{proposition}} 

\theoremstyle{definition}
\newtheorem{definition}[theorem]{Definition}
   \AddToHook{env/definition/begin}{\crefalias{theorem}{definition}} 

   \AddToHook{env/open/begin}{\crefalias{theorem}{open}} 
\newtheoremrep{remark}[theorem]{Remark}
   \AddToHook{env/remark/begin}{\crefalias{theorem}{remark}} 

\crefname{theorem}{Theorem}{Theorems}
\crefname{claim}{Claim}{Claims}
\crefname{corollary}{Corollary}{Corollaries}
\crefname{example}{Example}{Examples}
\crefname{lemma}{Lemma}{Lemmas}
\crefname{proposition}{Proposition}{Propositions}
\crefname{remark}{Remark}{Remarks}

\NewCommandCopy{\proofqedsymbol}{\qedsymbol}%
\newcommand{\remarkqedsymbol}{{$\triangle$}}%
\AtBeginEnvironment{proof}{\renewcommand{\qedsymbol}{\proofqedsymbol}}
\AtBeginEnvironment{nestedproof}{\renewcommand{\qedsymbol}{\nestedproofqedsymbol}}
\AtBeginEnvironment{example}{%
  \pushQED{\qed}\renewcommand{\qedsymbol}{\exampleqedsymbol}%
}
\AtEndEnvironment{example}{\popQED\endexample}
\AtBeginEnvironment{remark}{%
  \pushQED{\qed}\renewcommand{\qedsymbol}{\remarkqedsymbol}%
}
\AtEndEnvironment{remark}{\popQED\endexample}

\ifexternal
   \usetikzlibrary{external}
   \tikzexternaldisable %

\else
   
\fi
\ifdraft
    \usepackage[xcolor, hyperref, cleveref, notion, quotation, composition]{knowledge}
\else
    \usepackage[xcolor, hyperref, cleveref, notion, quotation, electronic]{knowledge}
\fi

\usepackage{mathcommand}
\knowledgeconfigure{quotation, protect quotation={tikzcd}}
\knowledgeconfigure{diagnose line=true, diagnose bar=true}

\definecolor{Dark Ruby Red}{HTML}{5d1416}
\definecolor{Dark Blue Sapphire}{HTML}{004c5c} %
\definecolor{Dark Gamboge}{HTML}{be7c00}

\IfKnowledgePaperModeTF{
    \knowledgestyle{notion}{color=black}
    \hypersetup{}
}{
    \knowledgestyle{intro notion}{color={Dark Ruby Red}, emphasize}
    \knowledgestyle{notion}{color={Dark Blue Sapphire}}
    \hypersetup{ %
        colorlinks=true,
        breaklinks=true,
        linkcolor={Dark Blue Sapphire}, %
        citecolor={Dark Blue Sapphire}, %
        filecolor={Dark Blue Sapphire}, %
        urlcolor={Dark Blue Sapphire},
    }
    \IfKnowledgeElectronicModeTF{
    }{
        \knowledgeconfigure{anchor point color={Dark Ruby Red}, anchor point shape=corner}
        \knowledgestyle{intro unknown}{color={Dark Gamboge}, emphasize}
        \knowledgestyle{intro unknown cont}{color={Dark Gamboge}, emphasize}
        \knowledgestyle{kl unknown}{color={Dark Gamboge}}
        \knowledgestyle{kl unknown cont}{color={Dark Gamboge}}
    }
}

\makeindex

\newcommand{\fakekl}{} %
\knowledgenewrobustcmd{\BPP}{\cmdkl{\mathsf{BPP}}}
\newrobustcmd{\coNP}{\fakekl{\mathsf{coNP}}}
\knowledgenewrobustcmd{\FP}{\cmdkl{\mathsf{FP}}} 
\knowledgenewrobustcmd{\FPsNP}{\cmdkl{\mathsf{FP}^{\mathsf{\#NP}}}} 
\knowledgenewrobustcmd{\FPsP}{\cmdkl{\mathsf{FP}^\mathsf{\#P}}}
\knowledgenewrobustcmd{\FPsPH}{\cmdkl{\mathsf{FP}^{\mathsf{\#PH}}}}
\knowledgenewrobustcmd{\NP}{\cmdkl{\mathsf{NP}}}
\knowledgenewrobustcmd{\NL}{\cmdkl{\ensuremath{\mathsf{NL}}}}
\knowledgenewrobustcmd{\logSpace}{\cmdkl{\mathsf{L}}}
\knowledgenewrobustcmd{\logCFL}{\ensuremath{\cmdkl{\mathsf{LogCFL}}}}
\knowledgenewrobustcmd{\PH}{\cmdkl{\mathsf{PH}}}
\knowledgenewrobustcmd{\PsP}{\cmdkl{\mathsf{P}^\mathsf{\#P}}}
\newrobustcmd{\Ptime}{\fakekl{\mathsf{P}}}
\knowledgenewrobustcmd{\sigmaptwo}{\cmdkl{\Sigma^p_2}}
\let\Sigtwo\sigmaptwo
\knowledgenewrobustcmd{\sNP}{\cmdkl{\mathsf{\#NP}}}
\knowledgenewrobustcmd{\sP}{\cmdkl{\mathsf{\#P}}} 
\knowledgenewrobustcmd{\sPH}{\cmdkl{\mathsf{\#PH}}}

\knowledgenewmathcommand{\numBipVerCov}{\cmdkl{\mathsf{\#BipVerCov}}}
\knowledgenewrobustcmd{\numBipIndep}{\cmdkl{\mathsf{\#BipIndepSet}}} %
\knowledgenewmathcommand{\numMinVerCov}{\cmdkl{\mathsf{\#MinVerCov}}}
\knowledgenewrobustcmd{\numPerMatch}{\cmdkl{\mathsf{\#PerfMatch}}} %
\knowledgenewrobustcmd{\numSTConn}{\cmdkl{\mathsf{\#stConnect}}} %
\knowledgenewmathcommand{\HamPathst}{\cmdkl{\mathsf{HamP_{s,t}}}}
\newmathcommand{\SAT}{\fakekl{\mathsf{SAT}}}
\knowledgenewmathcommand{\threeCol}{\cmdkl{\mathsf{3Col}}}
\knowledgenewmathcommand{\threeColbar}{\cmdkl{\overline{\mathsf{3Col}}}}
\knowledgenewmathcommand{\threeSAT}{\cmdkl{\mathsf{SAT}}}
\knowledgenewmathcommand{\twoQCNF}{\cmdkl{\mathsf{2QCNF}}}
\knowledgenewmathcommand{\VerCov}{\cmdkl{\mathsf{VerCov}}}
\knowledgenewmathcommand{\VerCovFun}{\cmdkl{\mathrm{VerCov}}}
\knowledgenewmathcommand{\IS}{\cmdkl{\mathrm{IS}}} %

\knowledgenewmathcommandPIE{\countMS}{%
   \cmdkl{\#}^{\cmdkl{\mathsf{ms}}}#1#2#3}
\knowledgenewmathcommandPIE{\countFMS}{%
   \cmdkl{\#}^{\cmdkl{\mathsf{fms}}}#1#2#3}
\knowledgenewmathcommandPIE{\countcMS}{%
   \cmdkl{\#}^{\cmdkl{\mathsf{cms}}}#1#2#3}
\knowledgenewmathcommandPIE{\countcFMS}{%
   \cmdkl{\#}^{\cmdkl{\mathsf{cfms}}}#1#2#3}
\knowledgenewmathcommandPIE{\countAns}{%
   \cmdkl{\#}^{\cmdkl{\mathsf{hom}}}#1#2#3}

\knowledgenewmathcommandPIE{\evalCountMS}{%
   \cmdkl{\textsc{eval-}\#}^{\cmdkl{\mathsf{ms}}}#1#2#3}
\knowledgenewmathcommandPIE{\evalCountFMS}{%
   \cmdkl{\textsc{eval-}\#}^{\cmdkl{\mathsf{fms}}}#1#2#3}
\knowledgenewmathcommandPIE{\evalCountcMS}{%
   \cmdkl{\textsc{eval-}\#}^{\cmdkl{\mathsf{cms}}}#1#2#3}
\knowledgenewmathcommandPIE{\evalCountcFMS}{%
   \cmdkl{\textsc{eval-}\#}^{\cmdkl{\mathsf{cfms}}}#1#2#3}
\knowledgenewmathcommandPIE{\evalCountAns}{%
   \cmdkl{\textsc{eval-}\#}^{\cmdkl{\mathsf{hom}}}#1#2#3}

\knowledgenewmathcommandPIE{\GIMC}{%
   \cmdkl{\mathsf{GIMC}#1#2#3}}
\knowledgenewmathcommandPIE{\FGIMC}{%
   \cmdkl{\mathsf{FGIMC}#1#2#3}}

\knowledgenewmathcommandPIE{\PQE}{%
   \cmdkl{\mathsf{PQE}#1#2#3}}
\knowledgenewmathcommandPIE{\SPQE}{%
   \cmdkl{\mathsf{SPQE}#1#2#3}}
\knowledgenewmathcommandPIE{\SPPQE}{%
   \cmdkl{\mathsf{SPPQE}#1#2#3}}
\knowledgenewmathcommandPIE{\PQEPhalf}{%
   \cmdkl{\mathsf{PQE}#1#2#3\!\left(\nicefrac 1 2\right)}}
\knowledgenewmathcommandPIE{\PQEPhalfOne}{%
   \cmdkl{\mathsf{PQE}#1#2#3\!\left(\nicefrac 1 2 ; 1\right)}}
\knowledgenewmathcommandPIE{\CPQE}{%
   \cmdkl{\mathsf{CPQE}#1#2#3}}
\knowledgenewmathcommandPIE{\CSPQE}{
   \cmdkl{\mathsf{CSPQE}#1#2#3}}
\knowledgenewmathcommandPIE{\CSPPQE}{
   \cmdkl{\mathsf{CSPPQE}#1#2#3}}
\knowledgenewmathcommandPIE{\CPQEPhalf}{
   \cmdkl{\mathsf{CPQE}#1#2#3\!\left(\nicefrac 1 2\right)}}
\knowledgenewmathcommandPIE{\CPQEPhalfOne}{
   \cmdkl{\mathsf{CPQE}#1#2#3\!\left(\nicefrac 1 2 ; 1\right)}}

\knowledgenewmathcommand{\DMH}{%
   \cmdkl{\mathsf{MinHom}^\mathsf{d}}}
\knowledgenewmathcommand{\DMHpath}{%
   \cmdkl{\mathsf{MinHom}^\mathsf{d}_\mathsf{path}}}
\knowledgenewmathcommand{\DMHself}{%
   \cmdkl{\mathsf{MinHom}^\mathsf{d}_\mathsf{self}}}
\knowledgenewmathcommand{\DUMH}{%
   \cmdkl{\mathsf{MinHom}^\mathsf{d/u}}}
\knowledgenewmathcommand{\DUMHpath}{%
   \cmdkl{\mathsf{MinHom}^\mathsf{d/u}_\mathsf{path}}}
\knowledgenewmathcommand{\DUMHself}{%
   \cmdkl{\mathsf{MinHom}^\mathsf{d/u}_\mathsf{self}}}
\knowledgenewmathcommandPIE{\Eval}{%
   \cmdkl{\mathsf{Eval}#1#2#3}}
\knowledgenewmathcommandPIE{\Rvlt}{%
   \cmdkl{\mathsf{Rvlt}#1#2#3}}
\knowledgenewmathcommand{\UMH}{%
   \cmdkl{\mathsf{MinHom}^\mathsf{u}}}
\knowledgenewmathcommand{\UMHpath}{%
   \cmdkl{\mathsf{MinHom}^\mathsf{u}_\mathsf{path}}}
\knowledgenewmathcommand{\UMHself}{%
   \cmdkl{\mathsf{MinHom}^\mathsf{u}_\mathsf{self}}}

\knowledgenewmathcommandPIE{\stShapley}{%
   \cmdkl{\mathsf{SVC}}^{\cmdkl{\star}}#1#2#3}

\knowledgenewmathcommandPIE{\dShapley}{%
   \cmdkl{\mathsf{SVC}}^{\cmdkl{\mathsf{dr}}}#1#2#3}
\newmathcommandPIE{\dShapleyNoKL}{%
   \mathsf{SVC}^{\mathsf{dr}}#1#2#3}
\knowledgenewmathcommandPIE{\dnShapley}{%
   \cmdkl{\mathsf{nSVC}}^{\cmdkl{\mathsf{dr}}}#1#2#3}
\knowledgenewmathcommandPIE{\dcShapley}{%
   \cmdkl{\mathsf{CSVC}}^{\cmdkl{\mathsf{dr}}}#1#2#3}
\knowledgenewmathcommandPIE{\dcnShapley}{%
   \cmdkl{\mathsf{nCSVC}}^{\cmdkl{\mathsf{dr}}}#1#2#3}

\knowledgenewmathcommandPIE{\mcShapley}{%
   \cmdkl{\mathsf{SVC}}^{\cmdkl{\mathsf{MC}}}#1#2#3}
\knowledgenewmathcommandPIE{\pShapley}{%
   \cmdkl{\mathsf{SVC}}^{\cmdkl{\mathsf{P}}}#1#2#3}
\knowledgenewmathcommandPIE{\rShapley}{%
   \cmdkl{\mathsf{SVC}}^{\cmdkl{\mathsf{R}}}#1#2#3}
\knowledgenewmathcommandPIE{\saShapley}{%
   \cmdkl{\mathsf{SVC}}^{\cmdkl{\mathsf{hom}}}#1#2#3}
\knowledgenewmathcommandPIE{\isShapley}{%
   \cmdkl{\mathsf{SVC}}^{\cmdkl{\mathsf{is}}}#1#2#3}

\newcommand{\minsupindex}{\mathsf{ms}}
\knowledgenewmathcommandPIE{\msShapley}{%
   \cmdkl{\mathsf{SVC}}^{\cmdkl{\minsupindex}}#1#2#3}
\knowledgenewmathcommandPIE{\msnShapley}{%
   \cmdkl{\mathsf{nSVC}}^{\cmdkl{\minsupindex}}#1#2#3}
\knowledgenewmathcommandPIE{\mscShapley}{%
   \cmdkl{\mathsf{CSVC}}^{\cmdkl{\minsupindex}}#1#2#3}
\knowledgenewmathcommandPIE{\mscnShapley}{%
   \cmdkl{\mathsf{nCSVC}}^{\cmdkl{\minsupindex}}#1#2#3}

\knowledgenewmathcommandPIE{\wShapley}{%
   \cmdkl{\mathsf{SVC}}#1#2#3}
\knowledgenewmathcommandPIE{\wnShapley}{
   \cmdkl{\mathsf{nSVC}}#1#2#3}
\knowledgenewmathcommandPIE{\wcShapley}{
   \cmdkl{\mathsf{CSVC}}#1#2#3}
\knowledgenewmathcommandPIE{\wcnShapley}{
   \cmdkl{\mathsf{nCSVC}}#1#2#3}
\knowledgenewmathcommandPIE{\sShapley}{%
   \cmdkl{\mathsf{SVC}}^{\cmdkl{\mathsf{s}}}#1#2#3}
\knowledgenewmathcommandPIE{\sharpShapley}{%
   \cmdkl{\mathsf{SVC}}^{\cmdkl{\mathsf{\#}}}#1#2#3}

\knowledgenewrobustcmd{\Sh}{\cmdkl{\mathrm{Sh}}} %
\knowledgenewrobustcmd{\Bz}{\cmdkl{\mathrm{Bz}}} %
\knowledgenewrobustcmd{\Slike}{\cmdkl{\S^c}} %

\knowledgenewmathcommandPIE{\scorefun}{%
   \cmdkl{\xi}#1#2#3}
\knowledgenewmathcommandPIE{\STscorefun}{%
   \cmdkl{\Xi}^{\cmdkl{\mathsf{\star}}}#1#2#3}
\knowledgenewmathcommandPIE{\stscorefun}{%
   \cmdkl{\xi}^{\cmdkl{\mathsf{\star}}}#1#2#3}

\knowledgenewmathcommandPIE{\onems}{%
   \cmdkl{\xi}#1#2#3}%
\knowledgenewmathcommandPIE{\SHAPscorefun}{%
   \cmdkl{\Xi}^{\cmdkl{\mathsf{SHAP}}}#1#2#3}
\knowledgenewmathcommandPIE{\shapscorefun}{%
   \cmdkl{\xi}^{\cmdkl{\mathsf{SHAP}}}#1#2#3}
\knowledgenewmathcommandPIE{\zetascorefun}{%
   \cmdkl{\zeta}#1#2#3}
\knowledgenewmathcommandPIE{\isscorefun}{%
   \cmdkl{\xi}^{\cmdkl{\mathsf{is}}}#1#2#3}

\knowledgenewmathcommandPIE{\Dscorefun}{%
   \cmdkl{\Xi}^{\cmdkl{\mathsf{dr}}}#1#2#3}
\knowledgenewmathcommandPIE{\dscorefun}{%
   \cmdkl{\xi}^{\cmdkl{\mathsf{dr}}}#1#2#3}

\knowledgenewmathcommandPIE{\MCscorefun}{%
   \cmdkl{\Xi}^{\cmdkl{\mathsf{MC}}}#1#2#3}
\knowledgenewmathcommandPIE{\mcscorefun}{%
   \cmdkl{\xi}^{\cmdkl{\mathsf{MC}}}#1#2#3}
\knowledgenewmathcommandPIE{\Pscorefun}{%
   \cmdkl{\Xi}^{\cmdkl{\mathsf{P}}}#1#2#3}
\knowledgenewmathcommandPIE{\pscorefun}{%
   \cmdkl{\xi}^{\cmdkl{\mathsf{P}}}#1#2#3}
\knowledgenewmathcommandPIE{\Rscorefun}{%
   \cmdkl{\Xi}^{\cmdkl{\mathsf{R}}}#1#2#3}
\knowledgenewmathcommandPIE{\rscorefun}{%
   \cmdkl{\xi}^{\cmdkl{\mathsf{R}}}#1#2#3}
\knowledgenewmathcommandPIE{\SAscorefun}{%
   \cmdkl{\Xi}^{\cmdkl{\mathsf{hom}}}#1#2#3}
\knowledgenewmathcommandPIE{\sascorefun}{%
   \cmdkl{\xi}^{\cmdkl{\mathsf{hom}}}#1#2#3}

\knowledgenewmathcommandPIE{\MSscorefun}{%
   \cmdkl{\Xi}^{\cmdkl{\minsupindex}}#1#2#3}
\knowledgenewmathcommandPIE{\msscorefun}{%
   \cmdkl{\xi}^{\cmdkl{\minsupindex}}#1#2#3}
\knowledgenewmathcommandPIE{\CMSscorefun}{%
   \cmdkl{\Xi}^{\cmdkl{\mathsf{cms}}}#1#2#3}
\knowledgenewmathcommandPIE{\cmsscorefun}{%
   \cmdkl{\xi}^{\cmdkl{\mathsf{cms}}}#1#2#3}

\knowledgenewmathcommandPIE{\Wscorefun}{%
   \cmdkl{\Xi}#1#2#3}
\knowledgenewmathcommandPIE{\wscorefun}{%
   \cmdkl{\xi}#1#2#3}
\knowledgenewmathcommandPIE{\Sscorefun}{%
   \cmdkl{\Xi}^{\cmdkl{\mathsf{s}}}#1#2#3}
\knowledgenewmathcommandPIE{\sscorefun}{%
   \cmdkl{\xi}^{\cmdkl{\mathsf{s}}}#1#2#3}
\knowledgenewmathcommandPIE{\SHARPscorefun}{%
   \cmdkl{\Xi}^{\cmdkl{\mathsf{\#}}}#1#2#3}
\knowledgenewmathcommandPIE{\sharpscorefun}{%
   \cmdkl{\xi}^{\cmdkl{\mathsf{\#}}}#1#2#3}

\knowledgenewrobustcmd{\sigmaless}[1]{\cmdkl{\sigma}_{\!\cmdkl{<}#1}}
\knowledgenewrobustcmd{\sigmaleq}[1]{\cmdkl{\sigma}_{\!\cmdkl{\le}#1}}

\knowledgenewrobustcmd{\bse}[1]{\cmdkl{X_{#1}}} %

\knowledgenewrobustcmd{\ann}{\cmdkl{\nu}}
\knowledgenewrobustcmd{\oneann}{\cmdkl{\mathbf{1}}}
\knowledgenewrobustcmd{\adom}{\cmdkl{\textit{adom}}} %

\knowledgenewmathcommandPIE{\Minsups}{%
   \cmdkl{\mathsf{MS}}#1#2#3}
\knowledgenewmathcommandPIE{\CMinsups}{%
   \cmdkl{\mathsf{CMS}}#1#2#3}

\knowledgenewmathcommand{\imodels}{%
   \cmdkl{\models_{\mathsf{i}}}}
\knowledgenewmathcommand{\nimodels}{%
   \cmdkl{\not\models_{\mathsf{i}}}}

\knowledgenewrobustcmd{\aC}{\cmdkl{C}} %
\knowledgenewrobustcmd{\aPi}{\cmdkl{\Pi}} %
\knowledgenewcommand{\Pic}{\cmdkl{\Pi C}} %

\knowledgenewmathcommandPIE{\IsubA}{%
   \cmdkl{\I#1#2#3}}

\knowledgenewrobustcmd{\cnames}{\cmdkl{\mathsf{N_C}}}
\knowledgenewrobustcmd{\dnames}{\cmdkl{\mathsf{N_D}}}
\knowledgenewrobustcmd{\inames}{\cmdkl{\mathsf{N_I}}}
\knowledgenewrobustcmd{\rnames}{\cmdkl{\mathsf{N_R}}}
\knowledgenewrobustcmd{\vnames}{\cmdkl{\mathsf{N_V}}}
\knowledgenewrobustcmd{\nulls}{\cmdkl{\mathsf{N_U}}}

\knowledgenewrobustcmd{\NRpm}{\cmdkl{\mathsf{N}_R^{\pm}}}

\knowledgenewrobustcmd{\terms}{\cmdkl{\mathsf{terms}}}
\knowledgenewrobustcmd{\mods}{\cmdkl{\mathsf{Mod}}}

\knowledgenewrobustcmd{\withT}[1]{(\T,#1)} %

\knowledgenewrobustcmd{\omqsat}{\mathrel{\cmdkl{\models}}} %

\knowledgenewrobustcmd{\anon}{\cmdkl{\mathsf{anon}}} %
\knowledgenewrobustcmd{\withanon}[1]{{#1}^{\cmdkl{\circ}}} %

\knowledgenewrobustcmd{\modelsmu}[1][\mu]{\cmdkl{\models}_{#1}} %

\knowledgenewrobustcmd{\dllitec}{\cmdkl{\ensuremath{\mathsf{DL\text{-}Lite}_{\mathsf{core}}}}}
\knowledgenewrobustcmd{\dlliter}{\cmdkl{\ensuremath{\mathsf{DL\text{-}Lite}_{\mathcal{R}}}}}
\knowledgenewrobustcmd{\dlliteh}{\cmdkl{\ensuremath{\mathsf{DL\text{-}Lite}_{\mathsf{Horn}}}}}

\knowledgenewrobustcmd{\EL}{\cmdkl{\mathcal{EL}}}
\knowledgenewrobustcmd{\ELI}{\cmdkl{\mathcal{ELI}}}
\knowledgenewrobustcmd{\elhibot}{\cmdkl{\mathcal{ELHI}_{\bot}}}
\knowledgenewrobustcmd{\elhifbot}{\cmdkl{\mathcal{ELHIF}_{\bot}}}

\knowledgenewrobustcmd{\Lmin}{\cmdkl{\L_{\min}}} %
\knowledgenewrobustcmd{\LminH}{\cmdkl{\L_{\sqcap}}} %

\newcommand{\subendo}{\mathsf{n}}
\newcommand{\subexo}{\mathsf{x}}
\knowledgenewrobustcmd{\Dn}[1][\D]{#1_{\cmdkl{\subendo}}}
\knowledgenewrobustcmd{\Dx}[1][\D]{#1_{\cmdkl{\subexo}}}

\knowledgenewrobustcmd{\constn}{\cmdkl{\textit{const}_{\cmdkl{\subendo}}}}
\knowledgenewrobustcmd{\constx}{\cmdkl{\textit{const}_{\cmdkl{\subendo}}}}

\knowledgenewrobustcmd{\atoms}{\cmdkl{\textit{atoms}}}
\knowledgenewrobustcmd{\arity}{\cmdkl{\mathrm{arity}}}
\knowledgenewrobustcmd{\const}{\cmdkl{\textit{const}}}
\knowledgenewrobustcmd{\Const}{\cmdkl{\mathsf{Const}}} %
\knowledgenewrobustcmd{\dom}{\cmdkl{\mathrm{dom}}} %
\knowledgenewrobustcmd{\mterms}{\cmdkl{\textit{terms}}}
\knowledgenewrobustcmd{\vars}{\cmdkl{\textit{vars}}}
\knowledgenewrobustcmd{\Var}{\cmdkl{\mathsf{Var}}}

\knowledgenewrobustcmd{\Unif}[1][q]{\cmdkl{\mathcal{M}}^{\textit{old}}_{#1}}
\knowledgenewrobustcmd{\Mergealt}[1][q]{\cmdkl{\mathcal{M}}_{#1}}

\knowledgenewrobustcmd{\vect}[1]{\cmdkl{\mathbf{#1}}} %
\knowledgenewrobustcmd{\1}[1]{\cmdkl{[}#1\cmdkl{]}} %

\knowledgenewrobustcmd{\lb}{\cmdkl{\mathrm{lb}}} %

\knowledgenewrobustcmd{\Esp}{\mathbf{E}} %
\knowledgenewrobustcmd{\Prob}{\mathbf{P}} %
\knowledgenewrobustcmd{\Vrnce}{\mathbf{V}} %

\knowledgenewrobustcmd{\partsof}[1]{\cmdkl{\mathcal{P}}(#1)} %
\knowledgenewrobustcmd{\fpartsof}[1]{\cmdkl{\mathcal{P}_{\mathsf{f}}}(#1)} %
\knowledgenewrobustcmd{\Perm}{\cmdkl{\matfrak{S}}} %
\knowledgenewrobustcmd{\Totord}{\cmdkl{\mathfrak{O}}} %

\knowledgenewrobustcmd{\homto}[1][]{\mathrel{\cmdkl{\xrightarrow{\smash{\textit{\tiny #1 \!hom}}}}}} %
\knowledgenewrobustcmd{\Chomto}[1][]{\mathrel{\cmdkl{\xrightarrow{\smash{\textit{\tiny #1 \!$\C$-hom}}}}}} %
\knowledgenewrobustcmd{\Pichomto}[1][]{\mathrel{\cmdkl{\xrightarrow{\smash{\textit{\tiny #1 \!$\Pic$-h}}}}}} %

\knowledgenewrobustcmd{\polyrx}{ %
   \mathrel{\cmdkl{\le_{\mathsf{P}}}}}
\knowledgenewrobustcmd{\polyeq}{ %
   \mathrel{\cmdkl{\equiv_{\mathsf{P}}}}}
\knowledgenewrobustcmd{\polyrxm}{ %
   \mathrel{\cmdkl{\le^{\mathsf{m}}_{\mathsf{P}}}}}
\knowledgenewrobustcmd{\polyeqm}{ %
   \mathrel{\cmdkl{\equiv^{\mathsf{m}}_{\mathsf{P}}}}}

\knowledgenewrobustcmd{\dcup}{\mathbin{\cmdkl{\uplus}}} %
\knowledgenewrobustcmd{\bigdcup}{\mathop{\cmdkl{\biguplus}}} %

\knowledgenewrobustcmd{\class}{\mathcal{C}} %

\knowledgenewmathcommand{\ACQ}{\cmdkl{\mathsf{ACQ}}} %
\knowledgenewmathcommand{\AQ}{\cmdkl{\mathsf{AQ}}} %
\knowledgenewmathcommand{\CQ}{\cmdkl{\mathsf{CQ}}} %
\knowledgenewmathcommand{\CQeq}{\cmdkl{\mathsf{CQ}^{=}}} %
\knowledgenewmathcommand{\CQeqneq}{\cmdkl{\mathsf{CQ}^{\neq,=}}} %
\knowledgenewmathcommand{\CQneg}{\cmdkl{\mathsf{CQ}^\lnot}} %
   \knowledgenewtextcommand{\CQneg}{\cmdkl{CQ$^\lnot$}} %
\knowledgenewmathcommand{\CQneq}{\cmdkl{\mathsf{CQ}^{\neq}}} %
\knowledgenewmathcommand{\CRPQ}{\cmdkl{\mathsf{CRPQ}}} %
\knowledgenewmathcommand{\FO}{\cmdkl{\mathsf{FO}}} %
\knowledgenewmathcommand{\IQ}{\cmdkl{\mathsf{IQ}}} %
\knowledgenewmathcommand{\RPQ}{\cmdkl{\mathsf{RPQ}}} %
\knowledgenewmathcommand{\sjfACQ}{\cmdkl{\mathsf{sjfACQ}}} %
\knowledgenewmathcommand{\sjfCQ}{\cmdkl{\mathsf{sjfCQ}}} %
\knowledgenewmathcommand{\UCQ}{\cmdkl{\mathsf{UCQ}}} %
\knowledgenewmathcommand{\UCQneg}{\cmdkl{\mathsf{UCQ}^\lnot}} %
   \knowledgenewtextcommand{\UCQneg}{\cmdkl{UCQ$^\lnot$}} %
\knowledgenewmathcommand{\UCQneq}{\cmdkl{\mathsf{UCQ}^\neq}} %
   \knowledgenewtextcommand{\UCQneq}{\cmdkl{UCQ$^\lneq$}} %
\knowledgenewmathcommand{\UCRPQ}{\cmdkl{\mathsf{UCRPQ}}} %
\knowledgenewmathcommand{\UsjfACQ}{\cmdkl{\mathsf{UsjfACQ}}} %

\knowledgenewrobustcmd{\lincq}{\cmdkl{\mathsf{LinCQ}}}
\knowledgenewrobustcmd{\sjflincq}{\cmdkl{\mathsf{sjf\text{-}LinCQ}}}
\knowledgenewrobustcmd{\sjfstarcq}{\cmdkl{\mathsf{sjf\text{-}StarCQ}}}

\knowledgenewmathcommandPIE{\onemsq}{%
   \cmdkl{q#2}#1#3}%

\knowledgenewmathcommand{\qphi}{\cmdkl{q_\Phi}} %
\knowledgenewmathcommand{\Dphi}{\cmdkl{\D_\Phi}} %

\knowledgenewrobustcmd{\Equivs}[1][q]{\cmdkl{\E}_{#1}}
\knowledgenewrobustcmd{\niceEquivs}[1][q]{\cmdkl{\widetilde\E}_{#1}}
\knowledgenewrobustcmd{\qE}[1][E]{\cmdkl{q_{#1}}}
\knowledgenewrobustcmd{\qEneq}[1][E]{\cmdkl{q_{#1}^{\neq}}}

\knowledgenewrobustcmd{\classSJk}[1][k]{\cmdkl{\mathcal{S\!J}_{\!#1}}}

\knowledgenewrobustcmd{\starRep}[2]{#1\mathbin{\cmdkl{\raisebox{0.2ex}{$\ast$}}}#2}%
\knowledgenewrobustcmd{\homstarRep}[2]{#1^{\cmdkl{\raisebox{0.1ex}{$\ast$}\!#2}}}%
\knowledgenewrobustcmd{\goodJ}{\cmdkl{\mathbf{J}}}

\knowledgenewmathcommand{\intatoms}{\cmdkl{\textit{int-atoms}}} %
\knowledgenewmathcommand{\frontiervars}{\cmdkl{\textit{shared-vars}}}
\knowledgenewmathcommand{\itw}{\cmdkl{\mathsf{iw}}}
\knowledgenewmathcommand{\tw}{\cmdkl{\mathsf{tw}}}
\knowledgenewmathcommand{\leaves}{\cmdkl{\mathsf{leaves}}}

\newcommand{\myparagraph}[1]{\smallskip\noindent\textbf{#1}~}
\newcommand{\aFact}{f}
\input{kl/knowledges.kl} %

\tikzset{n-core/.pic={
\coordinate (00) at (-2, 2);
\coordinate (01) at (-1, 2);
\coordinate (02) at (-0.49647, 2);
\coordinate (03) at (-1.5, 2.86603);
\coordinate (04) at (-2.5, 2.86603);
\coordinate (05) at (-3, 2);
\coordinate (06) at (-2.5, 1.13397);
\coordinate (07) at (-1.5, 1.13397);
\coordinate (08) at (-1.24824, 3.30209);
\coordinate (09) at (-2.75176, 3.30209);
\coordinate (010) at (-3.50353, 2);
\coordinate (011) at (-2.75176, 0.69791);
\coordinate (012) at (-1.24824, 0.69791);
\coordinate (013) at (-2, 1);
\coordinate (014) at (0, 0);
\coordinate (015) at (-3.75, 3.75);
\coordinate (016) at (-0.25, -0.25);
\coordinate (017) at (-2, 0);

\coordinate (-base) at (00);
\coordinate (-true) at (04);
\coordinate (-false) at (03);
\coordinate (-y) at (05);
\coordinate (-ybar) at (06);
\coordinate (-x) at (01);
\coordinate (-xbar) at (07);
\coordinate (-exit) at (014);
\coordinate (-nw) at (015);
\coordinate (-se) at (016);
\begin{pgfonlayer}{nodelayer}
\node [draw,circle] (-bnode) at (00) {};
\node [draw,circle, minimum height=1.5em] (-xnode) at (01) {};
\node [] (2) at (02) {};
\node [draw,circle, minimum height=1.5em] (-fnode) at (03) {};
\node [draw,circle, minimum height=1.5em] (-tnode) at (04) {};
\node [draw,circle, minimum height=1.5em] (5) at (05) {};
\node [draw,circle, minimum height=1.5em] (6) at (06) {};
\node [draw,circle, minimum height=1.5em] (7) at (07) {};
\node [] (10) at (010) {};
\node [] (11) at (011) {};
\node [] (12) at (012) {};
\node [] (13) at (013) {...};
\node [] (17) at (017) {$G_\phi$};
\end{pgfonlayer}
\begin{pgfonlayer}{edgelayer}
\draw [line width=.2em] (-bnode) to (-xnode);
\draw [line width=.2em] (-xnode) to (7);
\draw [line width=.2em] (7) to (-bnode);
\draw [line width=.2em] (-bnode) to (6);
\draw [line width=.2em] (6) to (5);
\draw [line width=.2em] (5) to (-bnode);
\draw [line width=.2em] (-bnode) to (-tnode);
\draw [line width=.2em] (-tnode) to (-fnode);
\draw [line width=.2em] (-fnode) to (-bnode);
\draw [line width=.2em] (-xnode) to node[sloped, rotate=90, pos=1.5] {...} (2.center);
\draw [line width=.2em] (7) to node[sloped, rotate=90, pos=1.5] {...} (12.center);
\draw [line width=.2em] (6) to node[sloped, rotate=90, pos=1.5] {...} (11.center);
\draw [line width=.2em] (5) to node[sloped, rotate=90, pos=1.5] {...} (10.center);
\draw [dashed, rounded corners] (015) rectangle (016);
\end{pgfonlayer}
}}

\begin{document}
\maketitle

\begin{abstract}

We consider the following fundamental problem: given a database $D$, 
Boolean conjunctive query (CQ) $q$, and fact $\aFact \in D$, decide whether 
$\aFact$ is \emph{relevant} to $q$ w.r.t. $D$, i.e., does $\aFact$ belong to a
minimal subset $S \subseteq D$ such that $S \models q$. 
Despite being of central importance to query answer explanation, 
the combined complexity of deciding query relevance 
has not been studied in detail, leaving open what makes this problem hard,
and which restrictions can yield lower complexity. Relevance has already been shown to 
be harder than query evaluation: namely, $\sigmaptwo$-complete for CQs,
even over a binary signature. We further observe that $\NP$-hardness applies 
already to %
(acyclic) chain CQs.
Our work identifies self-joins (multiple atoms with the same relation) 
as the culprit. Indeed, we prove that %
if we forbid or bound the occurrence of self-joins,  
then relevance has the same complexity as query evaluation, namely, $\NP$ (without structural restrictions) and 
$\logCFL$ (for bounded hypertreewidth classes). 
In the ontology setting, we establish an analogous result for  
ontology-mediated queries consisting of a CQ and $\dlliter$ ontology, 
namely that %
relevance is no harder than query answering
provided that we bound the interaction width (which generalizes both self-join width and a recently introduced `interaction-free' condition).
Our results thus pinpoint what makes relevance harder than query evaluation
and identify natural classes of queries which admit efficient relevance computation. 

\end{abstract}

\noindent
\raisebox{-.4ex}{\HandRight}\ \ This pdf contains internal links: clicking on a "notion@@notice" leads to its \AP ""definition@@notice"".%

\ifarxiv
  This is the long version of the KR'26 paper \cite{ourKR26}.
\else
  \AP A ""long version"" of this paper can be found at \url{https://arxiv.org/abs/XXXXXXXXXXX}.%
\fi

\section{Introduction}
There has been considerable interest in both the database and knowledge representation 
communities on devising methods for explaining why a given query answer holds, resulting in a diversity of approaches. 
In the database area, proposals include qualitative notions based upon causality \cite{DBLP:journals/pvldb/MeliouGMS11} and provenance \cite{GreenT17},
as well as quantitative notions of explanation based upon (causal or Shapley-value-based) responsibility measures, which involve assigning scores to facts based upon their contribution to making the query hold \cite{livshitsShapleyValueTuples2021,ourpods24,ourpods25paper}. 
In the setting of ontology-mediated query answering  (OMQA) \cite{poggiLinkingDataOntologies2008,bienvenuOntologyMediatedQueryAnswering2015,xiaoOntologyBasedDataAccess2018}, where query answers are computed by reasoning over the information contained in the data and the ontology, there have also been a diversity of explanation notions advanced in the literature. Some of these rely upon proofs which show step-by-step how a query answer is obtained \cite{DBLP:conf/otm/BorgidaCR08,DBLP:conf/ruleml/AlrabbaaBKK22}, others compute minimal subsets of facts that make the query hold \cite{DBLP:journals/jair/BienvenuBG19,DBLP:conf/ijcai/CeylanLMV19,DBLP:conf/ecai/CeylanLMV20,DBLP:journals/ai/CeylanLMV25}, while some recent work has begun to explore quantitative responsibility measures \cite{ourKR24,ourKR25}. 

In the present paper, our focus will be on explanations given as \emph{minimal supports} (\aka\ causes, MinEXes), 
i.e.\ subset-minimal sets of facts that suffice to obtain the query answer. %
More precisely, we shall be interested in the fundamental problem of \emph{query relevance}, 
which is to decide whether a given fact from the dataset is \emph{relevant}
to the query, in the sense that it belongs to at least one minimal support. 
Indeed, identifying the relevant facts can be considered a central task in query answer explanation, 
which is useful both for summarizing the minimal supports
and for aiding users in debugging unexpected query results. 
Relevance is also closely related to recently explored Shapley-based responsibility measures. Indeed, 
for classes of monotone queries such as conjunctive queries (with or without an ontology), 
a fact will be assigned a positive score iff it is relevant (cf.\ discussions in \cite{ourpods25paper}). 
Finally, it is worth noting that notions of relevance have also been considered 
for other kinds of explanations, including abductive explanations \cite{DBLP:journals/jacm/EiterG95}
justifications for ontology axiom entailment \cite{DBLP:journals/ai/PenalozaS17,DBLP:conf/esws/ChenMPY22}, 
explanations for query non-answers \cite{DBLP:journals/jair/CalvaneseOSS13}, and explanations for 
query (non)answers under various inconsistency-tolerant semantics \cite{DBLP:journals/jair/BienvenuBG19}. 

We briefly summarize here what is known about the complexity of the query relevance task, 
which has been explored for ontology-mediated queries (OMQs) formulated in 
description logic and existential rule ontology languages \cite{DBLP:conf/ijcai/CeylanLMV19,DBLP:conf/ecai/CeylanLMV20,DBLP:journals/ai/CeylanLMV25}. 
For data complexity, in which the size of the query and ontology are treated as fixed, 
the task is known to be tractable for (unions) of conjunctive queries (CQs), as well as for ontology-mediated queries that can be rewritten as UCQs (i.e.\ for so-called FO-rewritable ontology languages).
Indeed, this is an easy consequence of the fact that the size of minimal supports is bounded by a constant, 
making it possible to solve the relevance problem by reduction to first-order query evaluation. 
By contrast, for (ontology-mediated) query languages for which the size of minimal supports is unbounded, 
the relevance problem becomes $\NP$-hard in data complexity. For combined complexity, the relevance task has been proven 
$\sigmaptwo$-complete for conjunctive queries, even in the absence of an ontology and 
with only binary relations
\cite{DBLP:phd/ethos/Vaicenavicius20}. Additional combined complexity results have been established for various kinds of OMQs, 
with complexities ranging from $\sigmaptwo$ to $2\mathsf{EXPTIME}$, 
depending on the expressivity of the ontology language. 

While the preceding results show that relevance is generally harder than query evaluation,
to the best of our knowledge, there has not been any exploration of what precisely makes relevance harder than query evaluation,
nor the related question of which restrictions %
can lower the complexity. 
The purpose of our paper will thus be to pinpoint the sources of difficulty and identify classes of queries 
for which relevance is no harder than query evaluation, in particular, leading to the identification of tractable classes.  

\myparagraph{Contributions}
Our first contribution is to show that for conjunctive queries (without an ontology), 
it is the presence of self-joins (i.e.\ multiple query atoms using the same relation)
that make relevance harder than query evaluation. We start by observing that 
for `well-behaved' classes of conjunctive queries, like acyclic CQs and CQs of bounded hypertreewidth, 
for which query evaluation is tractable, the relevance problem is $\NP$-complete. 
In fact, relevance is $\NP$-hard even for acyclic chain CQs of the form 
$R(x_1,x_2) \land R(x_2,x_3) \land \dotsb \land R(x_{n-1},x_n)$ on a single relation (\Cref{th:cqbtw}),
which demonstrates that tractability cannot be regained by simplifying the query structure. 
However, if we restrict to CQs with `few' self-joins (via a notion of self-join width, recently introduced by \citeauthor{ourpods25paper} (\citeyear{ourpods25paper})), 
then the complexity of "relevance" matches that of "query evaluation",
namely $\NP$ for general CQs or $\logCFL$ if queries have bounded "hypertree width"
(\Cref{thm:boundedSJ-upper}). These results are summarized in Table \ref{tbl:db}. 

\newcommand{\shadedcell}{\cellcolor[HTML]{D0D0D0}}
\begin{table}[t]
\centering
\begin{tabular}{r|c|cc|}
\cline{2-4}
& \multirow{2}{*}{"Query evaluation"} & \multicolumn{2}{c|}{"Relevance"}\\ \cline{3-4} 
& & \multicolumn{1}{c|}{"sjw@self-join width"${<}\infty$} & "sjw@self-join width" $\infty$ \\ \hline
\multicolumn{1}{|l|}{"hw@hypertree width"${<}\infty$} & $\logCFL$-c & \multicolumn{1}{c|}{\shadedcell$\logCFL$-c}& \shadedcell$\NP$-c$^\dagger$\\ \hline
\multicolumn{1}{|l|}{"hw@hypertree width" $\infty$} & $\NP$-c & \multicolumn{1}{c|}{\shadedcell$\NP$-c} & $\Sigtwo$-c$^\ddagger$\\ \hline
\end{tabular}
\caption{Complexity results for classes of "conjunctive queries" with (un)bounded "hypertree width" ("hw@hypertree width") and "self-join width" ("sjw@self-join width"). Shaded cells indicate new results.
$\dagger$ $\NP$-hardness applies to "chain CQs" over a single binary relation. $\ddagger$ General $\Sigtwo$-hardness is known, we show it holds even for "CQs" with a single binary relation.}
\label{tbl:db}
\end{table}

\begin{table}[t]
\setlength{\tabcolsep}{4.1pt}
\begin{tabular}{l|cc|cc|}
\cline{2-5}
\multicolumn{1}{r|}{} & \multicolumn{2}{c|}{"OMQ evaluation"} & \multicolumn{2}{c|}{"Relevance"} \\ \cline{2-5} 
\multicolumn{1}{r|}{} & \multicolumn{1}{c|}{"iw@interaction width"${<}\infty$} & \multicolumn{1}{c|}{"iw@interaction width" $\infty$} & \multicolumn{1}{c|}{"iw@interaction width"${<}\infty$} & "iw@interaction width" $\infty$ \\ \hline
\multicolumn{1}{|l|}{"tw@treewidth"=1,  $\ell{<}\infty$} & \multicolumn{1}{c|}{\shadedcell$\NL$-c} & $\logCFL$-c & \multicolumn{1}{c|}{\shadedcell$\NL$-c} & \multicolumn{1}{c|}{\shadedcell$\NP$-c} \\ \hline
\multicolumn{1}{|l|}{"tw@treewidth" ${<}\infty$} & \multicolumn{1}{c|}{\shadedcell$\logCFL$-c} & $\NP$-c & \multicolumn{1}{c|}{\shadedcell$\logCFL$-c} & \multicolumn{1}{c|}{$\NP$-h}  \\ \hline
\multicolumn{1}{|l|}{"tw@treewidth" $\infty$} & \multicolumn{2}{c|}{$\NP$-c} & \multicolumn{1}{c|}{\shadedcell $\NP$-c} & $\Sigtwo$-c \\ \hline
\end{tabular}
\caption{Complexity results for classes of "ontology-mediated queries" in $(\dlliter, \CQ)$, with (un)bounded "treewidth" ("tw@treewidth"), "interaction width" ("iw@interaction width"),
or number of     ($\ell$) in "acyclic CQs" ("iw@interaction width"$=1$). Shaded cells indicate new results.
}
\label{tbl:omq}
\end{table}

We then move to considering the relevance problem for OMQs
consisting of a CQ and a description logic (DL) ontology. 
After observing that both conjunction ($\sqcap$) 
and qualified existential restrictions ($\exists R. C$) can make relevance hard even for 
atomic queries, we focus on ontologies formulated in $\dlliter$, a prominent 
lightweight DL often used for OMQA \cite{calvaneseetal:dllite}. We introduce 
a novel notion of interaction width, which can be seen as generalizing both self-join width 
and a recently introduced  `interaction-free' property \cite{ourKR25}. Our second main 
contribution is to show that for OMQs consisting of a CQ and $\dlliter$ ontology, 
bounding the interaction width makes relevance no harder than query evaluation.
The obtained results, summarized in \Cref{tbl:omq}, show that 
simultaneously bounding the interaction width and treewidth enables efficient ($\logCFL$)
relevance computation, and for acyclic bounded-leaf CQs, relevance drops to $\NL$.

Finally, as a third and last contribution, we investigate a more abstract notion of relevance based on graph homomorphisms, which underlies the "CQ" "relevance" problem. Concretely, we study the question of, given two graphs $G$ and $G'$, whether a given edge belongs to a minimal homomorphic image of $G$ in $G'$. We consider both directed and undirected variants of this problem, finding that they are both $\sigmaptwo$-complete, thereby providing natural and simple $\sigmaptwo$-complete problems of independent interest.
As a corollary, we obtain that $\sigmaptwo$-hardness of "CQ" "relevance" already holds over signatures consisting of a single binary relation.

\myparagraph{Organization} 
After preliminaries in \Cref{sec:prelims}, we introduce the relevance problem for "CQs" and "OMQs" in \Cref{sec:relpb} together with some basic observations and reductions.
We study the relevance problem for classes of "CQs" in \Cref{sec:relevance.cq}  and the relevance problem for "OMQs" in \Cref{sec:relevance.omq}. 
In \Cref{sec:graph}, we introduce and study the "Minimal Homomorphism Problem" for graphs.
We conclude in \Cref{sec:discussion} with a discussion of future work. 
\IfApx{}{Omitted proofs are in the "long version".}

\section{Preliminaries}\label{sec:prelims}
We introduce the key notions, terminology, and notation that will be used in this paper.
For detailed introductions to databases and description logics, we refer readers to \cite{abiteboul1995foundations,DBLP:books/daglib/0041477}. 

\myparagraph{Complexity}
We shall refer by \AP$\intro*\logSpace$ and $\intro*\NL$ the complexity classes of deterministic and nondeterministic logarithmic space, respectively. \AP$\intro*\NP$ is the 
class of decision problems solvable by a nondeterministic polynomial-time Turing machine, and 
\AP$\intro*\sigmaptwo$ the class of decision problems solvable by a nondeterministic polynomial-time Turing machine with access to an $\NP$ oracle.
The class \AP$\intro*\logCFL$ consists of all decision problems reducible to a context-free language, considered highly parallelizable, where
\[\logSpace \subseteq \NL \subseteq \logCFL \subseteq \mathsf{AC^1} \subseteq \mathsf{NC^2} \subseteq \Ptime \subseteq \NP \subseteq \sigmaptwo.\]

\myparagraph{Databases}
\AP
A (relational) ""database"" $D$  is a finite set of relational ""facts"" $R(\vect{a})$, where $R$ is a ""relation name"" of arity $k\geq 1$  and $\intro*\vect{a}$ is a $k$-ary vector of ""constants"". 
The ""signature"" of $D$ is the set of "relation names" that occur in the facts of $D$ together with their arity, and we speak of a ""binary signature"" if only unary and binary predicates are used.

\myparagraph{Conjunctive Queries} 
\AP ""Conjunctive queries"" ($\intro*{\CQ}$) are first-order formulas of the form $q(\vect x) = \exists \vect y ~ \alpha_1 \land \dotsb \land \alpha_n$ where the $\alpha_i$ are relational ""atoms"" that can contain "constants" and/or ""variables"", where $\vect{x}$ is the vector of ""free variables"" of the formula. By $\intro*\atoms(q)$ we denote $\set{\alpha_i}_i$.
We shall only be interested in ""Boolean"" "CQs", that is, "CQs" with no "free variables", and we will henceforth assume that all "CQs" are "Boolean". 

\AP
For any syntactic object $O$ (e.g., database, query), we will use $\intro*\vars(O)$ and $\intro*\const(O)$ 
to denote the sets of "variables" and "constants" contained in $O$, and let $\intro*\mterms(O)\defeq \vars(O) \cup \const(O)$ denote its set of ""terms"".

\AP
A ""homomorphism"" $q \intro*\homto D$ from a "CQ" $q$ to a "database" $D$ is a function $h: \mterms(q) %
 \to \const(D)$ such that $h(c)=c$ for every $c \in \const(q)$ and $R(h(t_1),\dotsc, h(t_n)) \in D$ for every "atom" $R(t_1, \dotsc, t_n)$ in $q$.
A \reintro{homomorphism} $h:q \homto q'$ between "CQs" is defined similarly. 
A "CQ" is a \AP""core"" if all the "homomorphisms" $h:q \homto q$ are injective. 
It is known that for any ("Boolean") "CQ" $q$,  $D \models q$  if{f} there is a homomorphism $h: q \homto D$, and that a "CQ" $q$ entails another "CQ" $q'$ if{f} $q' \homto q$. \AP We use $h(q)$ for the result of replacing each $x \in \vars(q)$ with $h(x)$ and call $h(q)$
a ""homomorphic image@@cq"" of $q$.

\myparagraph{Inequalities} A "CQ" $q$ with inequalities, or ""CQneq"", is a "CQ" extended with ""inequality atoms"" of the form $t \neq t'$ where $t,t'$ are "terms" (i.e., variables or constants). For any "database" or "CQ" $X$, the notion of "homomorphism" $h:q \homto X$ is restricted in the expected way, by further requiring that $h(t) \neq h(t')$ whenever $q$ contains the "inequality atom" $t \neq t'$.

\myparagraph{Evaluation}
\AP
For a class $\class$ of "CQs" (or of "CQneq"s), the ""query evaluation problem"" for $\class$ is the problem of, given a "database" $D$ and a "CQ" $q \in \class$, deciding whether $D \models q$. The problem is well known to be $\NP$-complete for the class of all "CQs" \cite{ChandraM77}, but it becomes tractable under structural restrictions. Except when explicitly stated otherwise, we always mean ""combined complexity"",
where both the query and database are part of the input (as compared to ""data complexity"", where only $D$ is considered as input). 

\myparagraph{Tree-like Queries}
\AP
An ""acyclic CQ"" is one in which there are no cycles in its underlying (hyper)graph. In the case of "binary signatures" this means that the underlying undirected graph, containing the edge $\{t,t'\}$ for each query atom $R(t,t')$, has no cycles. For "signatures" of bounded arity, the acyclicity is generalized to a measure of tree-likeness via the notion of ""treewidth"". We refer the reader to \cite[§3.1]{GottlobGLS16} for a definition. 
For unbounded signatures, several generalizations have been proposed, one of the most prominent being ""hypertree width""; in particular, "acyclic CQs" are precisely the queries of "hypertree width" 1.
We refer readers to \cite[Definition~3.1 \& following paragraph]{GottlobGLS16} for a definition. The "hypertree width" of a "CQneq" is defined to be that of the "CQ" obtained by considering $\neq$ as any ordinary binary relation.

\begin{theorem}[\cite{GottlobLS02}]\label{thm:boundedtw-logcfl}
  For every fixed $k>0$, the "query evaluation problem" for "CQneq"s of "hypertree width" at most $k$ is $\logCFL$-complete under log-space reductions. 
  The result holds also for "self-join free" "CQs".
  Further, checking if a "CQneq" has "hypertree width" at most $k$ is in $\logCFL$.\footnote{While the literature does not cover the class of "CQneq" strictly speaking, it can easily be seen that $\neq$-atoms can be handled in $\logCFL$ for the "query evaluation problem".}
\end{theorem}

\myparagraph{Description Logics}
\AP
A ""description logic"" ("DL") ""knowledge base"" ("KB") $\K =(\A,\T)$ consists of an "ABox" $\A$ and a "TBox" $\T$, 
constructed from mutually disjoint sets $\intro*\cnames$ of ""concept names"" (unary "relation names"), $\intro*\rnames$ of ""role names"" (binary relation names), 
and $\intro*\inames$ of ""individual names"" (here called "constants"). 
\AP An ""inverse role"" has the form $R^-$, with $R \in \rnames$,
and we use $\intro*\NRpm =  \rnames \cup \{R^- \mid R \in \rnames\}$ for the set of ""roles"". %
\AP
The ""ABox"" %
is a finite set of facts (i.e., a "database") over the "binary signature" given by $\cnames$ and $\rnames$.
\AP
The ""TBox"" (ontology) 
is a finite set of ""axioms"". Its
 form depends on the "DL" in question. 

\AP
Our results primarily concern lightweight DLs of the ""DL-Lite"" family \cite{calvaneseetal:dllite}.
We shall in particular consider the $\intro*{\dlliter}$ dialect, whose TBox
axioms take the form of \AP""concept inclusions"" $B \sqsubseteq C$
and \emph{role inclusions} $P \sqsubseteq S$,
built according to the following grammar
\begin{align*}
B := A \mid \exists P \quad C:= B \mid \neg B\quad
S:= P \mid \neg P
\end{align*}
where $A\in \cnames$ and $P \in \NRpm$. 
\AP
The logic $\intro*{\dllitec}$ is obtained from $\dlliter$ by disallowing role inclusions. 

\AP
Another prominent lightweight DL that we shall briefly mention is $\intro*{\EL}$,
where the
TBox consists of \AP""general concept inclusions"" ("GCIs") $D_1 \sqsubseteq D_2$
between concepts %
of the form:%
$$D:= \top \mid A \mid D \sqcap D \mid \exists R.D \qquad A\in \cnames, R\in \rnames  $$

\AP
The semantics of DL KBs is defined using 
""interpretations"" $\I=(\Delta^{\I},\cdot^{\I})$, 
 where the ""domain"" $\Delta^{\I}$ is a non-empty set 
 and %
 ${.}^{\I}$ maps each $a \in \inames$ to $a^{\I} \in \Delta^{\I}$, 
 each $A \in \cnames$ to $A^{\I} \subseteq \Delta^{\I}$, 
 and each $R \in \rnames$ to $R^{\I} \subseteq \Delta^{\I} \times \Delta^{\I}$.
 The function  $\cdot^{\I}$ is extended to complex concepts and roles: %
 $\top^{\I}=\Delta^{\I}$,
 $(\exists P)^{\I}= \{d \mid \exists e \in \Delta^{\I}, (d, e) \in P^{\I}\}$,
$(R^-)^\I=\{(e,d) \mid (d,e)\in R^\I \}$, $(C \sqcap D)^\I = C^\I \cap D^\I$. 
An interpretation $\I$ \emph{satisfies a fact} $A(a)$ (resp.\ $P(b, c)$) if $b^\I \in A^{\I}$ (resp.\ $(b^\I, c^\I) \in P^{\I}$). 
$\I$ \emph{satisfies a (concept or role) inclusion} $G \sqsubseteq H$ if $G^{\I} \subseteq H^{\I}$. 
\AP
We call  $\I$ a ""model of a KB"" $\K$, denoted $\I\models\K$, 
if $\I$ satisfies all axioms in $\T$ %
and all facts in $\A$. %
\AP
A KB $\K$ is ""consistent"" if it has a model. A KB $\K$ entails an inclusion or fact $\gamma$, written $\K \models \gamma$,
if every model $\I$ of $\K$ satisfies $\gamma$. Likewise, we write $\K \models \exists P (a)$ to mean $a^\I \in (\exists P)^\I$ in every model $\I$ of $\K$. 

\myparagraph{Ontology-Mediated Queries}
We say that a ("Boolean") "CQ" $q$ is \AP""entailed from a KB""~$\K=(\A, \T)$,
written $\K \models q$, if $\I \models q$ for every model $\I$ of $\K$. %
We may alternatively group $\T$ and $q$ together as a 
\AP ""ontology-mediated query"" (\reintro{OMQ}) $Q=(\T, q)$,
in which case we may write $\A \models Q$ to mean $(\A,\T) \models q$. 
The notation $(\mathcal{L}, \mathcal{Q})$ is used for the class of all OMQs $(\T, q)$ 
consisting of a TBox formulated in the DL $\L$
and a query $q \in \mathcal{Q}$. Aside from conjunctive queries ($\CQ$), we also consider 
classes of OMQs
whose component CQs are \AP""atomic queries"" ($\intro*{\AQ}$), i.e.\ CQs with a single atom. 

The \reintro{query evaluation problem for a class of OMQs} is defined analogously as before: 
given an "ABox" $\A$ and an OMQ $Q$ from the considered class, decide whether $\A \models Q$. 
The following theorem recalls some complexity results for query evaluation in subclasses of $(\dlliter,\CQ)$ (see \cite{DBLP:journals/jair/ArtaleCKZ09,DBLP:journals/jacm/BienvenuKKPZ18} and references therein). 

\begin{theorem}\label{omq-compl-eval}
Query evaluation for $(\dlliter,\CQ)$ is NP-complete. For fixed $m\geq 1$, $\ell>1$, 
the problem is $\logCFL$-complete for the classes obtained by restricting to 
$(\T,q)$ s.t.\  %
\begin{itemize}
\item %
$q$ is acyclic and has at most $\ell$ leaves, or %
\item %
$\T$ is a $\dllitec$ TBox and 
$q$ has treewidth at most $m$ %
\end{itemize}
Query evaluation is  $\NL$-complete for $(\dlliter,\AQ)$. 
\end{theorem}

\myparagraph{Canonical Model} Every "consistent"  $\dlliter$ "KB" $\K=(\A, \T)$ admits a so-called \AP ""canonical model"" $\intro*\IsubA_{\A,\T}$,
which maps homomorphically into every model of $\K$. As we shall use $\IsubA_{\A,\T}$ in some of our results, 
we recall here its definition. %
For the "domain" $\Delta^{\IsubA_{\A,\T}}$ of $\IsubA_{\A, \T}$, we use the set of all words $a P_{1} \ldots P_{n}$ ($n \geq 0$) such that $a \in \const(\A)$, $P_{i} \in \NRpm$, and: 
\begin{itemize}
\item if $n \geq 1$, then $(\A, \T) \models \exists P_1 (a)$ %
\item for $1 \leq i < n$, $\mathcal{T} \models \exists P_i^-
  \sqsubseteq \exists P_{i+1}$ and \mbox{$P_i^- \ne P_{i+1}$}. %
\end{itemize}
Elements in $\Delta^{\IsubA_{\A,\T}} \setminus \const(\A)$ will be 
called \AP""anonymous elements"". 
The interpretation function maps each $a \in  \const(\A)$ to itself and 
interprets concept and role names as follows: 
\begin{align*}
 A^{\IsubA_{\T,\A}} =   &   \,
 \{ a \in \const(\A) \mid (\A,\T) \models A(a) \} \cup  \\ &
 \{ a P_{1} \ldots P_{n}
\mid n \geq 1 \mbox{ and }
\T \models \exists{P_n^-} \sqsubseteq A \}{\color{green} )}  \\
 R^{\IsubA_{\T,\A}} =   &  \,
 \{ (a,b) \mid R(a,b) \in \A \} \, \cup  \\ & 
  \,   \{ (w_{1},w_{2})
\mid  
w_{2} = w_{1}P  \text{ and } \T \models P \sqsubseteq R \} \,\cup   \\
&   \,  \{ (w_{2},w_{1}) 
\mid \
w_{2} = w_{1}P  \text{ and } \T \models P \sqsubseteq R^{-} \}{\color{green} )} 
\end{align*}
Our results %
exploit the following well-known property: 

\begin{theorem}\label{canmod-prop} For every $\dlliter$ KB $(\A,\T)$, 
$(\A,\T) \models q$ iff $\IsubA_{\A,\T} \models q$ iff there is a "homomorphism" $q \homto \IsubA_{\A,\T}$. 
\end{theorem}

\smallskip

\section{Query Relevance Problem}\label{sec:relpb}

\begin{figure}[tb]
\centering
      \begin{tikzpicture}
\coordinate (00) at (-1.25, 0);
\coordinate (01) at (-0.5, 0);
\coordinate (02) at (1, 0);
\coordinate (03) at (2, 0);
\coordinate (04) at (4.5, 0);
\coordinate (05) at (3.25, 0);
\coordinate (06) at (-2, 0);
\coordinate (07) at (-2.5, 0);
\coordinate (08) at (0.5, 0);
\begin{pgfonlayer}{nodelayer}
\node [draw, circle] (0) at (00) {};
\node [draw, circle] (1) at (01) {};
\node [draw, circle] (2) at (02) {};
\node [draw, circle] (3) at (03) {};
\node [draw, circle] (4) at (04) {};
\node [draw, circle] (5) at (05) {};
\node [draw, circle] (6) at (06) {};
\node [] (7) at (07) {$q$:};
\node [] (8) at (08) {$D$:};
\end{pgfonlayer}
\begin{pgfonlayer}{edgelayer}
\draw [->,>=stealth] (6) to (0);
\draw [->,>=stealth] (0) to (1);
\draw [->,>=stealth] (2) to node[midway,above,color=blue] {$f_1$} (3);
\draw [->,>=stealth] (3) to node[pos=.4,above,color=blue] {$f_2$} (5);
\draw [->,>=stealth, in=135, out=45, loop] (5) to node[midway,above,color=blue] {$f_3$} ();
\draw [->,>=stealth] (4) to node[pos=.4,above,color=blue] {$f_4$} (5);
\end{pgfonlayer}
\end{tikzpicture}%
   
\caption{An example query and database to illustrate the notion of relevance. Arrows represent $R$ atoms/facts.}
\label{fig:relevance}
\end{figure}
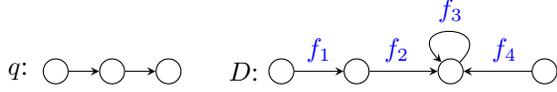

To improve the usability of information systems, it is important to be able 
to explain to users why a given answer $\vect{a}$ was obtained for the query $q$. 
Equivalently, we can rephrase this question in terms of Boolean queries, asking 
why the Boolean query $q(\vect{a})$ obtained by instantiating the 
free variables of the query with the answer tuple holds. 
Therefore, in what follows, we will focus w.l.o.g. on Boolean queries. 

For queries that are monotone (i.e., such that adding "facts" can never change the result from true to false), such as "CQs", %
a simple and natural way to explain why a query $q$ holds is to exhibit minimal subset(s) of the data that make $q$  true.

\begin{definition}
A \AP""support"" of a "CQ" $q$ in a "database" $D$ is any subset $S \subseteq D$ such that $S \models q$. 
It is a \AP""minimal support"" if further it does not strictly contain another "support".
\end{definition}

However, when a query has many "minimal supports", it may be more useful to first 
present users with the set of facts that occur in at least one "support". This is especially the case
if a query (answer) unexpectedly holds, in which case one will need to determine which 
facts are at fault. 

\begin{definition}
For a "CQ" $q$, "database" $D$ and fact $f \in D$, we say that $f$ is \AP""relevant"" to $q$ in $D$ if it belongs to some "minimal support" of $q$ in $D$.
\end{definition}

\begin{example}
Consider the CQ $q$ and database $D$ depicted in \Cref{fig:relevance}. 
There are two "minimal supports" for $q$ in $D$, which are $\{f_1,f_2\}$ and $\set{f_3}$. 
This shows that $f_1$, $f_2$, and $f_3$ are all "relevant". However, $f_4$ is not "relevant" 
despite appearing in a "homomorphic image@@cq" of $q$ in $D$, because the only "support" containing $f_4$
is $\{f_3, f_4\}$, which is not minimal. 
\end{example}

To define analogous notions for ontology-mediated queries, one must decide how to treat the case where the knowledge base is "inconsistent". 
Namely, do we want to consider minimal inconsistent subsets as supports for an OMQ? 
We believe it is more natural to separate debugging of inconsistencies from explaining query results, 
and so we shall require in our definitions that the KB is "consistent" 
(we shall discuss the case of "inconsistent" KBs in Section \ref{sec:discussion}). 

\begin{definition} Consider an "ontology-mediated query" $Q=(\T,q)$ with $\T$ a DL TBox and $q$ a "CQ" $q$, and an ABox $\A$ such that $(\A,\T)$ is "consistent". A \reintro{minimal support} for $Q$ in $\A$ is an inclusion-minimal subset $S \subseteq \A$ such that $S \models Q$,
and a "fact" $f \in \A$ is "relevant" to $Q$ in $\A$ if $f$ belongs to some "minimal support" of $Q$ in $\A$. 
\end{definition}

Our focus will be on analyzing the complexity of  deciding relevance for different classes of (ontology-mediated) "CQs". The decision problems are formally defined as follows. 
\AP\phantomintro{relevance problem}
\decisionproblem{\reintro{Relevance problem} for a class $\Q \subseteq \CQ$}%
                {Query $q\in \Q$, "database" $D$, and "fact" $f \in D$}%
                {Is $f$ "relevant" for $q$ in $D$?}

\decisionproblem{\reintro{Relevance problem} for "OMQ" class $(\L,\Q)$ \, ($\Q \subseteq \CQ$)}%
                {OMQ $(\T,q) \in (\L, \Q)$, ABox $\A$ such that $(\A,\T)$ is "consistent", and "fact" $f \in \A$}%
                {Is $f$ "relevant" for $(\T,q)$ in $\A$?}
\begin{remark}\label{rk:not-a-promise-pb}
   We treat these as standard decision problems in the complexity-theoretic sense, rather than `promise problems'. That is, an algorithm for $\mathcal{Q}$ must first verify that the input string is a valid encoding of a query from $\mathcal{Q}$.
\end{remark}

We conclude this section by stating some easy upper and lower bounds. 
First, we can observe that the "relevance problem" can be decided using a straightforward guess-and-check procedure:
guess a subset $S \subseteq D$ that contains the given fact $f$,
and verify that $S$ is a "minimal support". 
The latter can be checked by making repeated calls to a "query evaluation" oracle: we first check that the query holds in $S$, then check that the 
query no longer holds if any fact in $S$ is removed. We thus obtain the following generic 
upper bounds: %
\begin{proposition}\label{prop:rvltinsig2}
   For any class $\mathcal{Q}$ of (ontology-mediated) "CQs", the "relevance problem" for $\mathcal{Q}$ is in $\NP^{\mathsf{C}}$, where $\mathsf{C}$ is the complexity class of the "query evaluation problem" for $\mathcal{Q}$.
\end{proposition}

Regarding lower bounds, we cannot reasonably hope to achieve better complexity than for 
"query evaluation". Indeed, a query $q$ holds in a "database" $D$ 
just in the case that some "fact" in $D$ is "relevant" for $q$, which allows us to relate the 
complexities of the two problems as follows: 

\begin{propositionrep}\label{another-eval-easier}
 For any class $\mathcal{Q}$ of (ontology-mediated) "CQs", the  "query evaluation problem" for $\mathcal{Q}$ is in $\logSpace^{\mathsf{C}}$, where $\mathsf{C}$ is the complexity class of the "relevance problem" for $\mathcal{Q}$.
\end{propositionrep}
\begin{proof}
It suffices to iterate, in logspace, over all facts in the database and test whether the considered fact is "relevant". 
Query evaluation succeeds if{f} some fact is found to be "relevant". 
\end{proof}

While the preceding result can be used to show that the "relevance problem" is intractable
whenever "query evaluation" is, its formulation in terms of oracle calls is not the most convenient 
for establishing precise lower bounds. The following result provides a more direct (many-one) reduction: 

\begin{proposition}\label{prop:evaleasier}
Let $\Q$ be any class of "CQs" such that $q \in \Q$  implies that 
there exists $q \wedge \exists x. A(x) \in \Q$, for some relation $A$ that does not occur in $q$. 
Then there is a many-one logspace reduction from the "query evaluation problem" for $\Q$ %
to the "relevance problem" for $\Q$. An analogous statement holds for "OMQ" classes, except that we must further require that the relation $A$ does not occur in the TBox. %
\end{proposition}
\begin{proof}
Given $q \in \Q$ and a "database" $D$, let $A$ be any unary relation not appearing in $q$. Clearly, $q'=q \wedge \exists x. A(x)$ and 
$D' = D \cup \set{A(c)}$ can be constructed from $q, D$ using a logspace transducer. It then suffices to observe that 
$D \models q$ if{f} 
$A(c)$ is "relevant" to $q'$ in $D$. 
\end{proof}

The preceding proposition applies in particular to the whole class $\CQ$,
classes of "CQs" having "hypertree width" $< m$ (for any $m>0$), and 
classes of "acyclic CQs" with at most $\ell$ leaves. It does not however 
cover classes of "CQs" with a fixed finite "signature", nor the class of connected "CQs".

\section{Relevance for Conjunctive Queries}
\label{sec:relevance.cq}
The "relevance problem" for unrestricted "CQs" is known to be harder than the "evaluation problem", in fact one step higher in the polynomial hierarchy.

\begin{theorem}\label{th:cqsighard}\textup{\cite[Thm. 4.30]{DBLP:phd/ethos/Vaicenavicius20}\footnote{While the theorem statement in the 
    dissertation is on \emph{unions} of "CQs", the 
    proof makes use only of "CQs" over a "binary signature".}}
   The "relevance problem" for "CQs" is $\sigmaptwo$-complete, and $\sigmaptwo$-hardness holds already for "constant-free" "CQs" over a "binary signature".
\end{theorem}

We will see in the next subsections how two quantitative structural notions of queries affect this problem: the "hypertree width" and the amount of variables from different atoms that can be homomorphically merged into one fact.
\subsection{CQs with Bounded Treewidth}\label{ssec:cqbtw}
While "query evaluation" of "CQs" is $\NP$-hard, 
structural restrictions have been identified that make query evaluation tractable. %
This is in particular the case for any class of "CQs" with bounded "hypertree width" (\Cref{thm:boundedtw-logcfl}). 
By \Cref{prop:rvltinsig2}, this lowers the complexity of "relevance" to $\NP$. Interestingly, the $\NP$ lower bound applies even to (acyclic) \AP""chain CQs"" of the form $R(x_1,x_2) \land R(x_2,x_3) \land \dotsb \land R(x_{n-1},x_n)$ on a single binary relation $R$. 

\begin{theoremrep}\label{th:cqbtw}
For any fixed $k \geq 1$, the "relevance problem" for the class of all "CQs" of "hypertree width" at most $k$ is $\NP$-complete. Hardness holds even for the class of "chain CQs". 
\end{theoremrep}
\begin{proofsketch}
The "minimal supports" for the $n$-"atom" "chain CQ" $q_n$ are either simple cycles or simple paths of length $n$: indeed, any set strictly containing a cycle cannot be a "minimal support" since the cycle satisfies $q_n$.  We can then reduce from the existence of a Hamiltonian path on directed graphs. Given a directed graph $G$ with $n$ vertices we can produce a graph $G'$ by adding two new vertices $u,v$ and the edges $(u,v)$ and $(v,v')$ for every $v' \in V(G)$. Since $(u,v)$ is in no cycle, checking whether $(u,v)$ is "relevant" for $q_{n+1}$ in $G'$ is then equivalent to testing whether $G$ contains a simple path of length $n-1$ (i.e., a Hamiltonian path).
\end{proofsketch}
\begin{proof}
The fact that the problem is in $\NP$ is a direct consequence of \Cref{prop:rvltinsig2},
since the evaluation of "CQs" with bounded "hypertree width" is polynomial, see \Cref{thm:boundedtw-logcfl}.

\begin{figure}[tb]
\centering
      \begin{tikzpicture}
\coordinate (00) at (-1.25, 0);
\coordinate (01) at (0, 0);
\coordinate (02) at (2.75, 0);
\coordinate (03) at (0, 1.25);
\coordinate (04) at (2.75, -1);
\coordinate (05) at (1.5, 1);
\coordinate (06) at (1, 0.75);
\coordinate (07) at (2, -0.75);
\coordinate (08) at (1.5, 0.25);
\coordinate (09) at (2, 0.5);
\coordinate (010) at (-3, 1);
\coordinate (011) at (-2.5, 0);
\coordinate (012) at (4, 0);
\begin{pgfonlayer}{nodelayer}
\node [draw, circle, minimum height=1.5em] (0) at (00) {};
\node [draw, circle, minimum height=1.5em, fill=white] (1) at (01) {};
\node at (01) {$s$};
\node [draw, circle, minimum height=1.5em, fill=white] (2) at (02) {};
\node at (02) {$t$};
\node [] (5) at (05) {$G$};
\node [draw, circle, minimum height=1.5em] (6) at (06) {};
\node [draw, circle, minimum height=1.5em] (7) at (07) {};
\node [] (8) at (08) {...};
\node [draw, circle, minimum height=1.5em] (9) at (09) {};
\node [draw, circle, minimum height=1.5em] (12) at (012) {};
\end{pgfonlayer}
\begin{pgfonlayer}{edgelayer}
\draw [->,>=stealth] (0) to node[midway, above] {$e$} (1);
\draw [dashed, rounded corners] (03) rectangle (04);
\draw [->,>=stealth] (1) to (6);
\draw [->,>=stealth, bend right=15] (1) to (7);
\draw [->,>=stealth, bend right=15] (7) to (1);
\draw [->,>=stealth] (9) to (2);
\draw [->,>=stealth] (2) to (12);
\end{pgfonlayer}
\end{tikzpicture}%
   
\caption{Database $D_G$ for \Cref{th:cqbtw}. Arrows represent $r$ facts.}
\label{fig:hamil}
\end{figure}

Moving on to the $\NP$-hardness, we reduce from the $\AP\intro*\HamPathst$ problem of deciding, given a directed graph $G=(V,E)$ and distinguished vertices $s,t\in E$, if there exists a \AP""hamiltonian path"" from $s$ to $t$, that is a path $s \to \dots \to t$ that visits each vertex exactly once.%
Consider an input instance $G,s,t$, and denote by $n\defeq |V|$ the number of vertices in $G$. We build the database $D_G \defeq \{R(s',s),R(t,t')\}\cup\{R(u,v) \mid (u,v)\in E\}$, where $s',t'$ are fresh constants, as depicted in \Cref{fig:hamil},
and the query $q_G\defeq r^{n+1}(x,y)$. Denote by $e\defeq R(s',s)$ the first extra edge.

If there exists a "hamiltonian path" from $s$ to $t$, it forms a simple path $\P$ in $D_G$ of the form $r^{n-1}(s,t)$. $\P\not\models q$ because it is short of two edges, however $\P\cup\{e,R(t,t')\}$ is a "minimal support" for $q_G$, witnessing that $e$ is "relevant".

If there exists no "hamiltonian path" from $s$ to $t$, any "homomorphic image@@cq" of $q_G$ that contains $e$ will necessarily contain repeated vertices, hence a cycle. Now a "minimal support" for $q_G$ cannot strictly contain a cycle because any cycle satisfies $q$ alone. The fact $e$ is therefore "irrelevant".
\end{proof}

\subsection{CQs with Bounded Self-Join Width}
Observe that the "chain CQs" of \Cref{th:cqbtw} need unboundedly many "atoms" on the same relation name, or `"self-joins"'.
A \AP""self-join"" is a pair of "atoms" on the same "relation name", and a \AP""self-join free"" "CQ" is a "CQ" with no "self-joins".

We show that if the number of "self-join" "atoms" in "CQs" is bounded, then the complexity of the "relevance problem" 
matches that of "query evaluation", namely $\NP$ or even $\logCFL$ if queries have bounded "hypertree width" (such as "chain CQs"). 
Instead of counting the number of "self-joins", we use a more fine-grained notion of `"self-join width"', which will allow showing tractability for broad classes of "CQs".

For a tuple $\vect t$ or "atom" $\alpha=R(\vect t)$ we write $\vect t \AP\intro*\1i$ and $\alpha \reintro*\1i$ to denote the $i$-th element of the $\vect t$-tuple.
Two "atoms" $\alpha,\beta$ of a "CQ" $q$ are \AP""mergeable""  if there are two "homomorphisms" $h_\alpha : \alpha \homto D$, $h_\beta : \beta \homto D$
to an arbitrary "database" $D$  such that $h_\alpha(\alpha)=h_\beta(\beta)$ (in particular they must have the same "relation name"). 
An "atom" $\alpha$ is (individually) "mergeable" if it is "mergeable" with some other "atom" in $q$.
\begin{definition}[Self-join width]
    The \AP""self-join width"" of a "CQ" $q$ is the cardinality of the following set of variables
    \AP
    \begin{align*}
        \intro*\Mergealt \defeq \set{ \alpha\1 i : {}&\alpha, \alpha' \text{ are "mergeable" "atoms" of $q$} \\ &
        \text{  with } \alpha\1 i \neq \alpha'\1 i \text{ and } \alpha\1 i \in \vars(q) }.
    \end{align*}     
\end{definition}

This definition can be viewed as a generalization and improvement of %
a prior notion of "self-join width" introduced in \cite{ourpods25paper}%
\IfApx{, see \Cref{rk:old-vs-new-selfjoinwidth} for more details.}{.} Note that any "self-join free" query has "self-join width" 0, but the converse does not hold for "CQs" with constants.   
\begin{toappendix}
    \begin{remark}[On a prior definition of "self-join width"]\label{rk:old-vs-new-selfjoinwidth}
    A prior notion of \reintro{self-join width} for a "CQ" $q$ was introduced and studied in \cite{ourpods25paper}, which is defined as the size of
    \AP
    \begin{align*}
        \intro*\Unif \defeq \set{ t \in \mterms(\alpha) : \alpha \text{ is a "mergeable" "atom" of $q$}}.
    \end{align*}
    Observe how our current definition of "self-join width" is more general, in the sense that $\Mergealt \subseteq \Unif$.
    For example, the class $\class = \set{q_n : n > 0}$
    of Boolean "CQs" 
    \begin{align*}
        q_n \defeq \exists xx'y_1 \dotsc y_n ~ & R_n(x,y_1,\dotsc, y_n) \land {}\\
        & R_n(x',y_1,\dotsc, y_n) \land S(x,x')
    \end{align*}
    has unbounded "self-join width" with the definition above of \cite{ourpods25paper} since $\Unif[q_n] = \set{x,x',y_1, \dotsc, y_n}$. However, it has "self-join width" 2 with our definition, since $\Mergealt[q_n] = \set{x,x'}$.
    In this way, the alternative definition we propose yields larger classes of bounded "self-join width" and thus larger islands of tractability.
    \end{remark}    
\end{toappendix}
\begin{example}
The "CQ"  $q\defeq \exists xy. R(c,x) \land R(c',y)$, where $c,c'$ are distinct constants, %
has "self-join width" 0 since it has no two "mergeable" "atoms".
The "CQ" $q' \defeq \exists xyz ~ R(x,y) \land R(x,c) \land S(y,z)$ has "self-join width" $1$ as $\Mergealt[q'] = \set{y}$. 
\end{example}

For classes of bounded "self-join width" "CQs", the complexity of "relevance@relevance problem" matches that of "query evaluation@query evaluation problem":
\begin{theoremrep}\label{thm:boundedSJ-upper}
    For any fixed $k \geq 0, \ell > 0$:
    \begin{itemize}
        \item The "relevance problem" for "CQs" of "self-join width" at most $k$ is $\NP$-complete.
        \item The "relevance problem" for "CQs" of "self-join width" at most $k$ and "hypertree width" at most $\ell$ is $\logCFL$-complete.
    \end{itemize}
\end{theoremrep}
\begin{proof}
    We are given a "fact" $\aFact=R(\vect c)$, "database" $D$ and a "CQ" $q$.
    We first need to test if $q \in \classSJk$, which can be done in $\NL$ due to \Cref{lem:self-join-width-testing}-(a). 
    For the second statement, we have to further check that $q$ is of "hypertree width" at most $\ell$. This can be done in $\logCFL$ due to \Cref{thm:boundedtw-logcfl}.
    If $q$ does not meet the "self-join width" hypothesis (and the "hypertree width" hypothesis for the second statement), we reject the input (cf.\ \Cref{rk:not-a-promise-pb}).

    Otherwise, we compute the set $\Mergealt$ in $\logSpace$ via \Cref{lem:self-join-width-testing}-(b), and we then iterate, using logarithmic space, over all the possible $E \in \Equivs$ (of logarithmic size by \Cref{rk:constant-size-equivalences}) checking if there is one such $E$ meeting the criteria we describe next. 
    
    We first check that  $E \in \niceEquivs$ in $\logSpace$ by \Cref{lem:niceEquiv-exists}. We next check that there is some $R$-"atom" $R(\vect t)$ in $\qE$ \AP""consistent@@sjf"" with our input fact $R(\vect c)$, that is, such that 
    \begin{itemize}
        \item $\vect t\1i = \vect c\1i$ if $\vect t\1i$ is a constant,
        \item if $\vect t\1i = \vect t\1j$ then $\vect c\1i = \vect c\1j$.
    \end{itemize}
    If there is more than one such "consistent@@sjf" atom, then we iterate, using logarithmic space, over all of them testing if there is one meeting the properties that follow next.

    Consider now the result $\hat q$ of replacing in $\qEneq$ every $\vect t\1i$ which is a variable with the constant $\vect c\1i$. 
    We then have that any "homomorphic image@@cq" of $\hat q$ is a "minimal support" of $q$ containing $\aFact$.%
    Indeed, if $S$ is a "homomorphic image@@cq" of $\hat q$, there is a "homomorphism" $h: \hat q \homto S$ such that $h(\hat q)=S$. Extending $h$ to the variables of $\qE$ in the expected way (i.e., by mapping $\vect t \1i$ to $\vect c \1i$ for every variable $\vect t \1i$) then yields a "homomorphism" $h'$ such that $h': \qEneq \homto S$ and further $h'(\qE) = S$.
    Then, by \Cref{lem:niceEquiv-exists} (\ref{it:niceEquiv-exists:2} $\Rightarrow$ \ref{it:niceEquiv-exists:1}), $S$ is a "minimal support" of $q$.
    Further, if there exists a "minimal support" of $q$ containing $\aFact$, there must be some $E \in \niceEquivs$ and "atom" "consistent@@sjf" with $\aFact$ as described before.

    All these operations can be performed in logarithmic space since there are logspace transductions from $q$ to $\qEneq$, and from $\qEneq$ to $\hat q$, which can be composed.
    
    Observe that the "hypertree width" of $\hat q$ is at most that of $q$ plus $k$: collapsing variables does not increase the width and each added "inequality atom" $x \neq y$ may induce adding $x,y$ to some bags which increase the width by at most 1.

    We finally check $D \models \hat q$ by a call to the "query evaluation problem". This is in $\NP$ in general \cite[Theorem 7]{ChandraM77}, or in $\logCFL$ if we started with a class of bounded "hypertree width" by \Cref{thm:boundedtw-logcfl}.
    
    The $\NP$-hardness result follows via \Cref{prop:evaleasier} from the fact that "query evaluation" is already $\NP$-hard for "self-join free" "CQs".
    The $\logCFL$-hardness result follows via \Cref{prop:evaleasier} from the fact that "query evaluation" is $\logCFL$-hard for "self-join free" "acyclic" "CQs" (i.e., of "hypertree width" 1) by \Cref{thm:boundedtw-logcfl}.
\end{proof}

The rest %
of this section is devoted to proving this result.
Let us fix the class $\AP\intro*\classSJk$ of "CQs" of "self-join width" at most $k$.

\begin{lemmarep}\label{lem:self-join-width-testing}
    (a) Detecting if two given "atoms" are "mergeable" and more generally testing whether a "CQ" $q$ is in $\classSJk$ belongs to $\NL$. 
    (b) Computing $\Mergealt$ from a "CQ" $q \in \classSJk$ is in $\logSpace$.
\end{lemmarep}
\begin{proof}
    \proofcase{(a)}
    Observe that testing whether two "atoms" $R(t_1,\dotsc, t_n)$ and $R'(t'_1, \dotsc, t'_m)$  are "mergeable" is equivalent to testing that 
    (i) $R=R'$, 
    (ii) $n=m$, and 
    (iii) the graph $G = (V,E)$ has no two distinct constants in the same connected component, where $V = \set{t_1,\dotsc, t_n, t'_1,\dotsc, t'_m}$ and $E = \set{\set{t_i,t'_i} : i \in [n]}$.
    For testing (the negation of) (iii) one can non-deterministically guess a path between two distinct constants, and use the fact that $\NL$ is closed under complement by the Immerman-Szelepcsényi Theorem.

    For testing whether $q \in \classSJk$ we compute $\Mergealt$ as follows. We initialize a set $A =  \emptyset$. 
    Then, we iterate over every pair of "atoms" $\alpha,\alpha'$ in $q$, and if they are "mergeable" (which can be tested in $\NL$), we iterate over all indices $i$ of the arity of $\alpha$ and add $\alpha\1i$ to the set $A$ if $\alpha\1i$ is a variable such that $\alpha\1i \neq \alpha'\1i$.

    \proofcase{(b)}
    Note that the proof of $\NL$ from the previous point (a) for testing whether two "atoms" are "mergeable" becomes $\logSpace$ if we further assume that $q \in \classSJk$, since the lengths of paths in the graph is bounded by a constant (namely, $k$). Hence, the procedure to compute $\Mergealt$ of the previous paragraph becomes $\logSpace$.
\end{proof}

We shall need to consider all the different possible ways of `merging' two "mergeable" "atoms" by means of making some variables of $\Mergealt$ equal among them or among the "constants" contained in the query. For this, consider the set $\AP\intro*\Equivs$ of all equivalence relations over $\Mergealt \cup \const(q)$ such that no two distinct constants are in the same equivalence class. 
\begin{remark}\label{rk:constant-size-equivalences}
There is only a polynomial number of such equivalence relations $E \in \Equivs$ and further the cardinality $|E|$ of each equivalence relation is of constant size $\leq k$, under the suitable encoding where singleton classes are implicit, and hence $E$ can be stored in logarithmic space.
\end{remark}
For any equivalence relation $E \in \Equivs$ let $\AP\intro*\qE$ be the result of collapsing all the terms in the same equivalence class in $q$ and removing repeated "atoms". Let us further define $\AP\intro*\qEneq$ as the result of adding a "inequality atom" $t \neq t'$ to $\qE$ for each pair $(t,t') \in (\Mergealt \times (\Mergealt \cup \const(q))) \setminus E$. Intuitively, $\qEneq$ is the query stating that variables of $\Mergealt$ shall be collapsed\footnote{If the equivalence class $X$ contains a "constant" $c$, then all "variables" of $X$ are replaced by $c$; otherwise, we choose a representative "variable" $x \in X$ and replace all "variables" of $X$ by $x$.}  exactly according to $E$.

\begin{lemmarep}\label{lem:testing-qE-homs}
    The problems of testing, given $q \in \classSJk$ and two equivalence relations $E,E'$ from $\Equivs$, whether $\qE \homto \qE[E']$ holds, whether $\qEneq \homto \qEneq[E']$ holds, and whether $\qE$ is a "core", are all in $\logSpace$.
\end{lemmarep}
\begin{proof}
    If such $\qE \homto \qE[E']$ exists, then it must be identitary on all variables outside $\Mergealt$. There are constantly many remaining variables, and we can then go through all possible $\mathcal{O}(|\vars(\qE[E']) \cup const(\qE[E'])|^k)$ mappings for these in logarithmic space and checking whether some of these mappings yields a "homomorphism". The same argument applies for checking $\qEneq \homto \qEneq[E']$ by restricting the verification to mappings which are injective on $\Mergealt$.

    Finally, note that $\qE$ is not a "core" if, and only if, there is $E' \in \Equivs$ such that $\qE \homto \qE[E']$, $\qE[E'] \homto \qE$, and $\qE[E']$ has fewer "atoms" than $\qE$. Using only logarithmic space, we can go one by one all $E' \in \Equivs$ checking if there is one verifying these three conditions, in particular by using the algorithm for $\homto$ testing described above.
\end{proof}

The interest of the $\qEneq$ queries stems from the fact that, for a well-chosen subset of $\Equivs$, they completely characterize the "minimal supports" of $q$, as the next lemma shows. 
\begin{lemma}\label{lem:niceEquiv-exists}
For every $q \in \classSJk$ there exists a set $\AP\reintro*\niceEquivs \subseteq \Equivs$ of equivalence relations, such that the following statements are equivalent for every "database" $D$:
\begin{enumerate}
    \item \label{it:niceEquiv-exists:1} $S \subseteq D$ is a "minimal support" for $q$,
    \item \label{it:niceEquiv-exists:2} $S=h(\qE)$ for some $E \in \niceEquivs$ and "homomorphism" $h:\qEneq \homto D$. %
\end{enumerate}
Further, given an equivalence class $E \in \Equivs$, one can test in $\logSpace$ whether $E \in \niceEquivs$ (under the encoding of \Cref{rk:constant-size-equivalences}).
\end{lemma}
\begin{proof}
    The existence of $\niceEquivs$ essentially follows from the developments in \cite[proof of Theorem~6.6]{ourpods25paper}, but here we need to adapt definitions to have a better space complexity.
    Concretely, we define 
    \AP
    \begin{align*}
    \intro*\niceEquivs \defeq \set{ E \in \Equivs : {}& \qE \text{ is a "core", and for all $E' \in \Equivs$, }\\
    & \text{if $\qEneq[E'] \homto \qEneq$ then $\qEneq \homto \qEneq[E']$} }.
    \end{align*}

    \begin{claimrep}\label{cl:niceEquivs-satisfy-iff}
        $\niceEquivs$ satisfies the \ref{it:niceEquiv-exists:1} $\Leftrightarrow$ \ref{it:niceEquiv-exists:2} equivalence.
    \end{claimrep}
    \begin{proof}
        \proofcase{\ref{it:niceEquiv-exists:1} $\Rightarrow$ \ref{it:niceEquiv-exists:2}} 
        If $S \subseteq D$ is a "minimal support" for $q$, let $h: q \homto D$ be a "homomorphism" realizing it, i.e., such that $h(q) = S$.
        Consider the equivalence relation $E \in \Equivs$ \AP""induced by@@equivrel"" $h$ (i.e., where $(t,t') \in E$ if{f} $h(t)= h(t')$ for all $t,t' \in \Mergealt \cup \const(q)$).
        Hence, $\hat h: \qEneq \homto S$ via the relativization $\hat h$ of $h$ onto the variables of $\qEneq$. Observe that $\hat h : \qEneq \homto S$ \AP""atom injective"", in the sense that the size of $S$ equals the number of "atoms" of $\qE$\IfApx{ (cf.\ \Cref{cl:induced-size})}{}.
        We need to show that $E \in \niceEquivs$. 
        If $\qE$ is not a "core", then it is easy to see that $S$ cannot be a "minimal support"; let us hence assume that $\qE$ is a "core". 
        By contradiction, suppose that $g:\qEneq[E'] \homto \qEneq$ 
        but $\qEneq \not\homto \qEneq[E']$ for some $E' \in \Equivs$. 
        This means that $g(\qE[E'])$ does not contain all "atoms" of $\qE$: let $\alpha$ be one of these missing "atoms". See \Cref{fig:boundedsjw} for an illustration.
        Note that $h(\alpha) \in S$ cannot be in $h(g(\qE[E']))$ since $h$ "atom injective"; hence $h(g(\qE[E'])) \subsetneq S$ contradicting the fact that $S$ is a "minimal support" of $q$.
        \begin{figure}[tb]
      \begin{tikzpicture}
\coordinate (00) at (0.25, -0.25);
\coordinate (01) at (-1.25, 1.75);
\coordinate (02) at (-1, 1);
\coordinate (03) at (0, 0);
\coordinate (04) at (-0.5, 1.5);
\coordinate (05) at (-3.25, 1.5);
\coordinate (06) at (-2.25, 0.5);
\coordinate (07) at (-2.25, 1);
\coordinate (08) at (-1, 0.5);
\coordinate (09) at (-2.5, 1.5);
\coordinate (010) at (-0.75, 1.75);
\coordinate (011) at (4.5, -0.25);
\coordinate (012) at (2.5, 1.75);
\coordinate (013) at (2.75, 1);
\coordinate (014) at (4.25, 0);
\coordinate (015) at (3.5, 1.5);
\coordinate (016) at (2.75, 0.5);
\coordinate (017) at (3.5, 2);
\coordinate (018) at (0.5, 0.75);
\coordinate (019) at (2.25, 0.75);
\coordinate (020) at (-3, 1.75);
\coordinate (021) at (-0.25, 2);
\begin{pgfonlayer}{nodelayer}
\node [] (0) at (00) {};
\node [] (1) at (01) {};
\node [] (2) at ($(02)+(-.125,0)$) {};
\node [] (3) at ($(03)+(.125,-.125)$) {};
\node [color=red, draw] (4) at (04) {\footnotesize $\alpha$};
\node [] (5) at (05) {};
\node [] (6) at (06) {};
\node [] (7) at (07) {};
\node [] (8) at ($(08)+(-.125,0)$) {};
\node [] (9) at (09) {};
\node [] (10) at (010) {};
\node [] (11) at (011) {};
\node [] (12) at (012) {};
\node [] (13) at ($(013)+(-.125,0)$) {};
\node [] (14) at ($(014)+(.125,-.125)$) {};
\node [color=red, draw] (15) at ($(015)+(0,-.1)$) {\footnotesize $h(\alpha)$};
\node [] (16) at (016) {};
\node [] (17) at (017) {$h(q_E)=S$};
\node [] (18) at (018) {};
\node [] (19) at (019) {};
\node [color=blue] (20) at (020) {$q_{E'}^{\neq}$};
\node [] (21) at (021) {$q_{E'}^{\neq}$};
\end{pgfonlayer}
\begin{pgfonlayer}{edgelayer}
\draw [rounded corners] (1.center) rectangle (0.center);
\draw [color=blue, rounded corners] (2.center) rectangle node[midway] {$g(q_{E'})$} (3.center);
\draw [color=blue, rounded corners] (5.center) rectangle (6.center);
\draw [rounded corners] (12.center) rectangle (11.center);
\draw [color=blue, rounded corners] (13.center) rectangle node[midway] {$h(g(q_{E'}))$} (14.center);
\draw [->,>=stealth] (18.center) to node[midway, yshift=2em] {\footnotesize (atom injective)} node[midway, above] {\footnotesize hom} node[midway, below] {\footnotesize $h$} (19.center);
\draw [->,>=stealth, bend right=75] (10.center) to node[midway] {/} (9.center);
\draw [color=blue,->,>=stealth] (7.center) to node[sloped, midway, above] {\footnotesize hom} node[sloped, midway, below] {\footnotesize $g$} (8.center);
\end{pgfonlayer}
\end{tikzpicture}%
   
            \caption{Visual support for proof of \Cref{lem:niceEquiv-exists}.}
            \label{fig:boundedsjw}
        \end{figure}

        \smallskip

        \proofcase{\ref{it:niceEquiv-exists:2} $\Rightarrow$ \ref{it:niceEquiv-exists:1}} 
        By means of contradiction, suppose $S$ is the $h$-image for $h : \qEneq \homto D$ (i.e., $h(\qE) = S$) where $E \in \niceEquivs$, but $S$ is not minimal since there is some $g: q \homto S'$ with $S' \subsetneq S$.
        Let $E' \in \Equivs$ be the equivalence relation "induced by@@equivrel" $g$, so that $f(\qE[E'])=S'$ for some $f: \qEneq[E'] \homto D$ -- as before, note that $f$ "atom injective": $|\atoms(\qE[E'])| = |S'|$\IfApx{ (cf.\ \Cref{cl:induced-size})}{}.
        Since $f(\qE[E']) \subseteq h(\qE)$, it follows that $\qEneq[E'] \homto \qEneq$\IfApx{ (cf.\ \Cref{cl:compose-homs})}{}, and hence that, since $E \in \niceEquivs$, we must also have $\qEneq \homto \qEneq[E']$.    
        By composing the "homomorphisms" $\qEneq \homto \qEneq[E']$ and $f : \qEneq[E'] \homto S'$ we obtain that $|S'|$ is the number of "atoms" of $\qE$\IfApx{ (by \Cref{cl:induced-size})}{}.
        But this is in contradiction with what we have already established, namely, that the number of "atoms" of $\qE$ is $|S|$ via the "atom injective" "homomorphism" $h$, and that $|S|\neq |S'|$ since $S' \subsetneq S$.    
    \end{proof}

    \begin{toappendix}
        \begin{claim}[taken from {\cite[Claim~50]{ourpods25arxivV3}}]
            \label{cl:induced-size}
        If a homomorphism $h: \qEneq \homto D$ is so that $h(\qE)=S$, then $|S|$ is the number of "atoms" of $\qE$
        \end{claim}
        \begin{nestedproof}
            The proof can be found in the cited paper.
        \end{nestedproof}
        \begin{claim}[taken from {\cite[Claim~54]{ourpods25arxivV3}}]
            \label{cl:compose-homs}
            If for $E,E' \in \Equivs$ there are "homomorphisms" $h:\qEneq \homto D$ and $h':\qEneq[E'] \homto D$ with $h'(\qE[E']) \subsetneq h(\qE)$, then there is a "homomorphism" $\qEneq[E'] \homto \qEneq$.
        \end{claim}
        \begin{nestedproof}
            The proof can be found in the cited paper.
        \end{nestedproof}
    \end{toappendix}

    Since the cardinalities of the equivalence classes in $\Equivs$ and $\niceEquivs$ are bounded by a constant (cf.\ \Cref{rk:constant-size-equivalences}), they can be stored and manipulated using only logarithmic space. The complexity statement then follows from \Cref{lem:testing-qE-homs}.
\end{proof}

With the previous lemma in place, \Cref{thm:boundedSJ-upper} follows:
\begin{proof}[Proof sketch of \Cref{thm:boundedSJ-upper}]
    We are given a "fact" $\aFact=R(\vect c)$, "database" $D$ and a "CQ" $q$.
    We first need to test if  $q$ meets the hypothesis. We test $q \in \classSJk$ in $\NL$ due to \Cref{lem:self-join-width-testing}-(a), and for the second statement we test "hypertree width" in $\logCFL$ due to \Cref{thm:boundedtw-logcfl}.

    We now compute the set $\Mergealt$ in $\logSpace$ via \Cref{lem:self-join-width-testing}-(b), and we then iterate in logspace over all $E \in \Equivs$  checking if there is one such $E$ meeting the following criteria. 
    We first check that  $E \in \niceEquivs$ in $\logSpace$ by \Cref{lem:niceEquiv-exists}. We next check that there is some $R$-"atom" $R(\vect t)$ in $\qE$ `"consistent@@sjf"' with our input fact $R(\vect c)$, that is, so that $R(\vect t) \homto R(\vect c)$.
    Consider now the result $\hat q$ of replacing in $\qEneq$ every $\vect t\1i$ which is a variable with the constant $\vect c\1i$. Observe that, due to \Cref{lem:niceEquiv-exists}, any "homomorphic image@@cq" of $\hat q$ is a "minimal support" of $q$ containing $\aFact$, and further if there exists a "minimal support" of $q$ containing $\aFact$, there must be some $E \in \niceEquivs$ and "atom" "consistent@@sjf" with $\aFact$ as described before.
    All these operations can be performed in logarithmic space since there are logspace transductions from $q$ to $\qEneq$, and from $\qEneq$ to $\hat q$, which can be composed.%
    Note that the "hypertree width" of $\hat q$ is at most that of $q$ plus $k$.

    We finally check $D \models \hat q$ by a call to the "query evaluation problem". This is in $\NP$ in general \cite[Theorem 7]{ChandraM77}, or in $\logCFL$ if we started with a class of bounded "hypertree width" by \Cref{thm:boundedtw-logcfl}.
    
    The lower bounds hold by \Cref{prop:evaleasier} due to equivalent bounds for "query evaluation" of "self-join free" "CQs".
\end{proof}

\section{Relevance for Ontology-Mediated Queries}\label{sec:relevance.omq}
Now that we have a clear picture of the combined complexity
of "relevance" for different classes of "CQs", we shall push further 
and consider "ontology-mediated queries" $(\T, q)$ where $\T$ is a description logic "ontology" and $q$ is a "CQ". 
We will again be interested in understanding how the complexity of the
"relevance" task compares to that of "query evaluation". 

\subsection{New Sources of Hardness}\label{sec:sources}
The addition of an "ontology" introduces new sources of hardness for the "relevance problem",
which make the problem difficult even when restricted to "atomic queries". Indeed, 
a first source of hardness stems from the ability to capture reachability in the data,
which in DLs can be done using qualified existential restrictions, %
present in $\mathcal{EL}$ and its extensions. This makes relevance 
intractable even in data complexity: 

\begin{proposition}[\cite{DBLP:conf/ecai/CeylanLMV20}]\label{prop:elhard}
"Relevance" for %
$(\L,\AQ)$ is $\NP$-hard in  data complexity (hence also in combined complexity)
for any DL $\L$ that %
can express $\exists R. A \sqsubseteq A$.
\end{proposition}
\begin{proof} %
Any directed graph $G$ can be represented as an "ABox" $\A_G$ which contains %
$R(v_1,v_2)$ for each directed edge $(v_1,v_2)$. %
Then, to decide whether an edge $(v_1,v_2)$ in $G$ lies on a simple path from $s$ to $t$ (an $\NP$-hard problem, cf.\ \cite[Lemma 5.3]{khalilComplexityShapleyValue2023}),
it suffices to check whether $R(v_1,v_2)$ is "relevant" for the "OMQ" $(\{\exists R. A \sqsubseteq A\}, A(s))$  
w.r.t.\  the "ABox" $\A_G \cup \{A(t)\}$. 
\end{proof}
\noindent Observe  that the preceding result precludes the possibility of obtaining tractability for classes of OMQs defined by bounding any combination of width notions or any other parameters that assign a finite value to each OMQ.

In fact, $\mathcal{EL}$ contains a further source of hardness: concept conjunction. 
Indeed, if the ontology language can express propositional definite Horn clauses, %
then deciding "relevance" is hard (in combined complexity) even for "atomic queries": %

\begin{proposition}\label{prop:omqhard} %
Relevance for the "OMQ" class $(\L,\AQ)$ is $\NP$-hard for any "DL" $\L$ with concept conjunction and "GCIs".
\end{proposition}
\begin{proof}
We reduce from $\SAT$. Consider a "CNF" formula $\phi \defeq \bigwedge_{j=1}^m c_j$ with variables in $V\defeq\{v_1 \dots v_n\}$.
We build the "AQ" $q\defeq A(d)$, the "ABox" $\A\defeq\{X(d),P_i(d),N_i(d) \mid 1 \leq i \leq n\}$, and the %
"TBox" $\T\defeq \T_t \cup \T_p \cup \T_n \cup \T_c$ with:
\[%
\begin{array}{ll}
\T_t : \{P_i \sqcap N_i \ic A \mid i\in [n]\} &\T_p : \{P_i \ic C_j \mid v_i \in c_j\}\\
\T_c : \{X\sqcap\bigsqcap_{j=1}^m C_j \ic A\} & \T_n : \{N_i \ic C_j \mid \overline{v}_i \in c_j\}
\end{array}%
\]
It can be verified that the
fact $X(d)$ is "relevant" to $(\T,q)$ w.r.t.\ $\A$  iff $\phi$ is satisfiable.
\end{proof}
\noindent This implies that  if the "DL" allows for concept conjunction, then we cannot hope to identify tractable classes of "OMQs" by bounding parameters which are dominated by the query size, since NP-hardness holds already for "atomic queries". 

In view of these results, the only "DLs" %
for which we can hope to obtain tractability 
results by imposing conditions on the query 
are logics that contain neither conjunction nor admit %
qualified existential restrictions on the left-hand-side of axioms. 
This naturally leads us to explore core fragments of the "DL-Lite" family, which verify these requirements. 

\subsection{General Results for DL-Lite}
We shall henceforth focus on "OMQs" $(\T,q)$ whose "TBox" $\T$ is formulated in $\dlliter$. %
We consider $\dlliter$ since it is a well-known core dialect (notably underlying the OWL 2 QL profile)
for which the combined complexity of "OMQ evaluation" has been well explored \cite{DBLP:journals/jacm/BienvenuKKPZ18}.  

First, we determine the combined complexity of testing "relevance" for arbitrary "OMQs" in $(\dlliter, \CQ)$,
obtaining the same complexity as for "CQs" in databases:

\begin{proposition}\label{prop:omqsighard}
The "relevance problem" for $(\dlliter, \CQ)$ is $\Sigtwo$-complete.
\end{proposition}
\begin{proof}
The lower bound directly follows from \Cref{th:cqsighard}, as the hardness proof only uses binary relations. 
The upper bound follows from \Cref{prop:rvltinsig2} and $\NP$ membership for "query evaluation" in $(\dlliter, \CQ)$, see \Cref{omq-compl-eval}. 
\end{proof}

Next, we pinpoint the complexity of well-behaved "OMQ" classes whose "evaluation problem" has been 
previously shown to be tractable (specifically, $\logCFL$-complete): 

\begin{proposition}
The "relevance problem" is $\NP$-complete for the following classes of "OMQs": %
\begin{itemize}
\item class of "OMQs" $(\T,q) \in (\dlliter, \CQ)$ such that $q$ is acyclic and has at most $\ell$ leaves (for any fixed $\ell>1$)
\item class of "OMQs" $(\T,q) \in (\dllitec, \CQ)$ such that $q$ has treewidth at most $m$ (for any fixed $m\geq 1$)
\end{itemize}
\end{proposition}
\begin{proof}
The $\NP$ lower bounds follows from \Cref{th:cqbtw}. For the upper bounds,
we combine \Cref{prop:rvltinsig2} with existing $\logCFL$ results for "OMQ" answering in the considered classes, recalled in  \Cref{omq-compl-eval}. 
\end{proof}

The preceding results show that, just as in the plain database setting, 
the worst-case complexity of "relevance" is one level higher than "query evaluation". 
Inspired by the positive impact of restricting "self-joins", we shall next explore
how an analogous notion can be employed to obtain lower complexities
for deciding "relevance" of "OMQs". 

\subsection{Bounded Interaction Width OMQs}
The most obvious way of translating the notion of self-join-free queries 
to the OMQA setting would be to consider "OMQs" $(\T,q)$ where $q$
is a self-join-free "CQ". However, the resulting notion does not have the 
desired properties due to interactions between atoms that arise from the ontology:

\begin{proposition}\label{prop:sjfomq}
The hardness of \Cref{prop:omqsighard} holds with the restriction that the component "CQs" are "self-join free".
\end{proposition}
\begin{proof}
   We reduce from the relevance problem of "CQs" with relations of arity 2. Let $q$ be such a query, and $D$ a database on the same signature, which we can see as an "ABox" since it does not have any fact of arity $> 2$. We remove all "self-joins" in $q$ by replacing each instance of a relation $R$ with a fresh $R_i$, along with the axiom $R_i \ic R$ that makes $R_i$ a particular instance of a $R$. Since $D$ does not contain any of those $R_i$, the resulting "OMQ" will behave exactly as $q$ on $D$ and its subsets, hence the equivalence between the two when it comes to the "relevance problem".
\end{proof}

This lead \citeauthor{ourKR25} (\citeyear{ourKR25}) to define a notion of "interaction-free" "OMQ", 
whose purpose is to ensure that an "ABox" fact can only be 
used to satisfy a single atom of the query (in a single way).  
We recall below the definition of interaction-free "OMQs"\footnote{We reformulated slightly the original definition to suit our purposes, but it yields the same notion of interaction-free OMQ.}. 

\begin{definition}[Interacting query atoms]\label{inter-atoms}
Given an "OMQ" $Q=(\T,q) \in (\dlliter, \CQ)$, we say that distinct 
atoms $\alpha,\beta$ of $q$ ""interact"" %
if there exists 
a fact $f$ such that  $(\{f\},\T)\models \alpha$ and $(\{f\},\T)\models \beta$ (with $\alpha$ and $\beta$ treated as Boolean "CQs"). 
An atom $\alpha$ \reintro{interacts} with itself if there exist a fact $f$,  %
"homomorphisms" $h_1: \alpha \homto \IsubA_{\{f\},\T}$ and $h_2: \alpha \homto \IsubA_{\{f\},\T}$, and variable $x\in \vars(\alpha)$ 
such that $h_1(x) \in \const(f)$ and $h_2(x) \neq h_1(x)$.
We denote by $\intro*\intatoms(Q)$ the set of atoms in $Q$ that (self-)"interact". %
\end{definition}

\begin{definition}[Interaction-free OMQ]\label{int-free}
An "OMQ" $Q \in (\dlliter, \CQ)$ %
is ""interaction-free"" if %
$\intatoms(Q)= \emptyset$. 
\end{definition}

\begin{figure}[tb]
\centering
      \begin{tikzpicture}
\coordinate (00) at (-3, 0);
\coordinate (01) at (-1.75, 0);
\coordinate (02) at (-0.5, 0);
\coordinate (03) at (1, 0);
\coordinate (04) at (2.25, 0);
\coordinate (05) at (3.5, 0);
\coordinate (07) at (0, -1);
\coordinate (08) at (-3.5, 0);
\coordinate (09) at (0.5, 0);
\begin{pgfonlayer}{nodelayer}
\node [draw, circle, minimum height=1.5em] (0) at (00) {};
\node at (00) {$x$};
\node [label={$A$}, draw, circle, minimum height=1.5em] (1) at (01) {};
\node at (01) {$y$};
\node [draw, circle, minimum height=1.5em] (2) at (02) {};
\node at (02) {$z$};
\node [draw, circle, minimum height=1.5em] (3) at (03) {};
\node [label={$A$}, draw, circle, minimum height=1.5em] (4) at (04) {};
\node at (04) {$c$};
\node [draw, circle, minimum height=1.5em] (5) at (05) {};
\node (5a) at ($(05)+(1.2em,0)$) {$A$};
\node at (05) {$d$};
\node [below, color=blue] at (5a) {\scriptsize $f_3$};
\node [] (8) at (08) {$q$:};
\node [] (9) at (09) {$\A$:};
\end{pgfonlayer}
\begin{pgfonlayer}{edgelayer}
\draw [->,>=stealth] (0) to node[midway,below,color=blue] {\scriptsize $\alpha$} node[midway,above] {$R$} (1);
\draw [->,>=stealth] (1) to node[midway,below,color=blue] {\scriptsize $\beta$} node[midway,above] {$R$} (2);
\draw [->,>=stealth] (4) to node[midway,above] {$R$} (3);
\draw [->,>=stealth] (4) to node[midway,above] {$R'$} node[midway,below,color=blue] {\scriptsize $f_1$} (5);
\draw [->,>=stealth, in=135, out=45, loop] (5) to node[midway,above] {$R$} node[pos=.6,xshift=-1ex,yshift=1ex,color=blue] {\scriptsize $f_2$} ();
\end{pgfonlayer}
\end{tikzpicture}%
   
\caption{Interacting atoms where $\T\defeq\{R'\ic R, R'\ic R^-\}$.}
\label{fig:interaction}
\end{figure}

\begin{example}\label{ex:interaction}
   Consider the KB $(\A,\T)$ and CQ $q$ depicted in \Cref{fig:interaction}. $\alpha$ "interacts" with $\beta$ via $f_1$, but both also interact with themselves by the two distinct homomorphisms that map to $(c,d)$ in either direction, since $\IsubA_{\{f_1\},\T} \supseteq c \xleftrightarrow{R} d$. This is especially important because, out of $(x,y)\mapsto (c,d)$ and  $(x,y)\mapsto (d,c)$, only the latter witnesses the "relevance" of $f_1$: the former can only be extended with $z\mapsto d$, whose image strictly contains the "minimal support" $\{f_2,f_3\}$.
\end{example}

We propose to generalize "interaction-free" "OMQs" by introducing the notion of "interaction width":

\begin{definition}[Interaction width]
\AP
The ""interaction width"" of an "OMQ" $Q \in (\dlliter, \CQ)$ %
is $|\intatoms(Q)|$.
\end{definition}

Observe that, as expected, "interaction-free" "OMQs" have "interaction width" zero. 
It is also worth noting that, differently from "self-join width", we count the number of "atoms" in $\intatoms(Q)$, not the number of "variables"  in such "atoms" (but this difference is insignificant on "binary signatures").

\begin{example}
The OMQ $Q$ from \Cref{fig:interaction} has "interaction width" 2, 
since $\intatoms(Q)= \{\alpha, \beta\}$. The OMQ $Q'_n=(\T_n,q_n)$ where $q_n \defeq \bigwedge_{i=1}^n R_i(x_{i-1},x_i)$ and $\T\defeq\{R_i \ic R\}$ has "interaction width" $n$ since $\alpha \homto \IsubA_{\{R(c,d)\}}$ for every atom $\alpha$ of $q$, hence $\intatoms(Q'_n)= \atoms(Q'_n)$.
\end{example}

The remainder of the section will be dedicated to establishing 
the following theorem, which shows that by bounding the
"interaction width", we can obtain the same complexity for 
"query relevance" as for "query evaluation". 

\begin{theorem}\label{omq-thm}
The "relevance problem" is in $\NP$ for every subclass of $(\dlliter, \CQ)$ 
having bounded "interaction width". $\logCFL$ (resp.\ $\NL$) membership 
holds if we further require component "CQs" have bounded "treewidth" %
(resp.\ the component "CQs" are "acyclic" with bounded number of leaves). 
\end{theorem}

In what follows, we suppose that we have an input "OMQ" $Q=(\T,q)$ of "interaction width" $k$
and an "ABox" $\A$. %
\AP
We shall call a fact $f \in \A$ ""potentially relevant"" to atom $\alpha$ in $(\T,q)$
if $(\{f\},\T)\models \alpha$. The next lemma distinguishes two kinds of "potentially relevant" facts: %

\begin{lemmarep}\label{pot-relevant-cases}
If a fact $f \in \A$ is "relevant" to $Q=(\T,q)$ w.r.t.\ $\A$, then $f$ is "potentially relevant" to either
\begin{enumerate*}[(i)]
\item\label{pot-relevant-casesa} a single "atom" in $\atoms(q) \setminus \intatoms(Q)$ (and no other atom), or
\item\label{pot-relevant-casesb} one or more "atoms" in $\intatoms(Q)$.  
\end{enumerate*}
\end{lemmarep}
\begin{proof}
Suppose that $f \in \A$ is "relevant" to $Q=(\T,q)$ w.r.t.~$\A$. This means that there is a 
"minimal support" $S \subseteq \A$ for $(\T,q)$ 
that contains $f$. There is thus a homomorphism $h: q \homto \IsubA_{S,\T}$, but 
minimality implies that $q$ does not embed homomorphically into $\IsubA_{S \setminus \{f\},\T}$. 
Due to the "canonical model" 
definition, there must thus exist some atom $\alpha_f \in q$ such that 
$\alpha_f \homto \IsubA_{\{f\},\T}$, i.e.\ $f$ is "potentially relevant" to $\alpha_f$.
First suppose that $\alpha_f \in \atoms(q) \setminus \intatoms(Q)$. 
Then by definition, there cannot exist another atom $\beta_f$ such that 
$\beta_f \homto \IsubA_{\{f\},\T}$, so $\alpha_f$ is the unique atom for which $f$
is "potentially relevant". Otherwise, $\alpha_f \in \intatoms(Q)$, hence $f$ is "potentially relevant"
for at least one atom in $\intatoms(Q)$, as in condition \ref{pot-relevant-casesb}. 
\end{proof}

We start by testing whether the given fact $f \in \A$ is "potentially relevant"
to some atom of $q$, and return no if not. This can be done in $\NL$ by guessing a query atom
and performing atomic "query evaluation". %
We may thus focus on testing "relevance" for 
"potentially relevant" facts of types \ref{pot-relevant-casesa} or \ref{pot-relevant-casesb}. 

Consider first the case where $f$ is "potentially relevant" to a single "atom" $\alpha_f \in \atoms(q) \setminus \intatoms(\T,q)$.
Then there exists %
$h: \alpha_f \homto \IsubA_{\{f\},\T}$, and moreover, every such homomorphism agrees 
on which variables are sent to which "ABox" constants (and which are mapped to the anonymous part). %
Define $q_{\text{-}\alpha_f}$ as the "CQ" obtained from $q$ by removing $\alpha_f$ and replacing variable $x$ by $h(x)$ if $x$ occurs both in 
$\alpha_f$ and another atom of $q$. The following lemma provides a direct reduction of "relevance" to "query evaluation": 

\begin{lemmarep}\label{type-i}
Let $f \in \A$ be "potentially relevant" to $\alpha_f \in \atoms(q) \setminus \intatoms(\T,q)$. %
Then $f$ is "relevant" to $(\T,q)$ w.r.t.\ $\A$ iff $(\A,\T) \models q_{\text{-}\alpha_f}$. %
\end{lemmarep}
\begin{proof}
First suppose that $f$ is "relevant" to $(\T,q)$ w.r.t.\ $\A$, and let $S$ be a minimal support for $(\T,q)$ in $\A$
that contains $f$. It follows that there is a homomorphism $h: q \homto \IsubA_{S,\T}$. Moreover, since 
$f$ is "potentially relevant" to $\alpha_f \in \atoms(q) \setminus \intatoms(\T,q)$, 
there is a witnessing homomorphism $h': \alpha_f \homto \IsubA_{\{f\},\T}$. Since $\alpha_f \in \atoms(q) \setminus \intatoms(\T,q)$,
we know that if $x \in \vars(\alpha_f)$ and $h'(x) \in \const(\A)$,  then $h'(x)=h(x)$. It follows that $h$
also witnesses that $q_{\text{-}\alpha_f} \homto \IsubA_{S,\T}$, which yields $(S,\T) \models q_{\text{-}\alpha_f}$, 
hence $(\A,\T) \models q_{\text{-}\alpha_f}$.

For the other direction, suppose that $(\A,\T) \models q_{\text{-}\alpha_f}$, and let 
$S$ be a minimal support for $(\T,q_{\text{-}\alpha_f})$ in $\A$. We claim that $S \cup \{f\}$
is a minimal support for $(\T,q)$ in $\A$: 
\begin{itemize}
\item We first show that $S \cup \{f\}$ is a support for $(\T,q)$ in $\A$. 
Indeed, we know that there is a homomorphism 
$h: q_{\text{-}\alpha_f}\homto \IsubA_{S,\T}$ and a second homomorphism $h': \alpha_f \homto \IsubA_{\{f\},\T}$
witnessing that $f$ is "potentially relevant" for $\alpha_f$. Now let $h''$ be defined by setting 
$h''(x)=h(x)$ if $x \in \vars(q_{\text{-}\alpha_f})$
and $h''(x)=h'(x)$ if $x \in \vars(\alpha_f)$ (note that $h''$ is well defined as $\vars(q_{\text{-}\alpha_f}) \cap \vars(\alpha_f)=\emptyset$). 
Then by appealing to the canonical model construction, we can see that $h''$ defines a homomorphism 
$q \homto \IsubA_{S \cup \{f\},\T}$. It follows that $(S,\T) \models q$, i.e.\  $S$ is a support for $(\T,q)$.
\item Now we establish minimality. Let us suppose for a contradiction that $S \cup \{f\}$ is not a minimal support, and let $S' \subsetneq S \cup \{f\}$
be a stricter smaller support. If $f \not \in S'$, then $q \homto \IsubA_{S,\T}$. By appealing to the canonical model construction, 
there must exist $f' \in S$ such that $\alpha_f \homto \IsubA_{\{f'\},\T}$, which contradicts our assumption that $\alpha_f \not \in \intatoms(\T,q)$. 
Thus, it must be the case that $S' \subsetneq S$ and $q \homto \IsubA_{S' \cup \{f\},\T}$. 
However, since $S$ is a minimal support for $(\T,q_{\text{-}\alpha_f})$, 
either $f$ is being used to satisfy another atom than $\alpha_f$, 
or there is a homomorphism witnessing $q \homto \IsubA_{S' \cup \{f\},\T}$ that sends $x \in \vars(\alpha_f)$ to 
a constant different from $h'(x)$.
Both options are disallowed since $\alpha_f \not \in \intatoms(\T,q)$, yielding the desired contradiction. \qedhere
\end{itemize}
\end{proof}

It remains to consider the more challenging case, where $f$ is "potentially relevant" to 
some %
atom(s) in $\intatoms(Q)$. The basic idea will be to iterate over subsets $S \subseteq \A$ of at most $k$ facts which 
contain $f$ and make true the subquery given by $\intatoms(Q)$ and can be extended to build a minimal support for $q$.
The following lemma makes precise which properties of $S$ to check in order to conclude that $f$ is "relevant". It refers to the set 
\AP$\intro*\frontiervars(\T,q)$ of variables that occur both in $\intatoms(Q)$ and in
$\atoms(q) \setminus \intatoms(\T,q)$. 

\begin{toappendix}
\begin{lemma}\label{lem:sharedinabox}
Let $x\in \frontiervars(\T,q)$ and $h: q \homto \IsubA_{\A,\T}$. Then $h(x)\in \const(\A)$.
\end{lemma}
\begin{proof}
Assume for a contradiction that $h(x)$ is an "anonymous constant".
By construction of the "canonical model" in $\dlliter$, any "anonymous element" $w= aP_1\dots P_n$ of $\IsubA_{\A,\T}$ can be generated by a fact $f_w\in \A$ (\ie\ $w\in \IsubA_{\{f_w\},\T}$) and this $f_w$ will also generate all intermediate $aP_1\dots P_k$ with $1\leq k<n$ as well as all facts of $\IsubA_{\A,\T}$ over these elements. In particular, if we take $\alpha\in\intatoms(Q)$ and $\beta\in\atoms(q) \setminus \intatoms(\T,q)$ that contain $x$ (as per the definition of $\frontiervars(\T,q)$), and denote by $w$ the "anonymous constant" of $h(\{\alpha,\beta\})$ that is the furthest into the "anonymous" region, then we will have such an $f_w$ \st\  $h(\{\alpha,\beta\})\inc \IsubA_{\{f_w\},\T}$. This constitutes an "interaction" between $\alpha$ and $\beta$ which contradicts $\beta\in\atoms(q) \setminus \intatoms(\T,q)$.
\end{proof}
\end{toappendix}

\begin{lemmarep}\label{lem:potrel123}
Let $f \in \A$ be "potentially relevant" to at least one atom in $\intatoms(\T,q)$.
Then $f$ is "relevant" to $(\T,q)$ w.r.t.\ $\A$ iff there exists a subset $S \subseteq \A$ with $f \in S$
and $|S| \leq k$ and $h_S: \intatoms(Q) \homto \IsubA_{S,\T}$ such that:
\begin{enumerate}
\item\label{lem:potrel123a} if $x \in \frontiervars(\T,q)$, then $h_S(x)\in \const(\A)$
\item\label{lem:potrel123b} $(\A,\T) \models q_{h_S}$ where $q_{h_S}$ is obtained by removing all atoms in $\intatoms(\T,q)$
and replacing $x \in \frontiervars(\T,q)$ by $h_S(x)$
\item\label{lem:potrel123c} there is no $S' \subsetneq S$ and $h_{S'}: \intatoms(Q) \homto \IsubA_{S',\T}$
such that $h_{S'}(x)=h_{S}(x)$ for every $x \in \frontiervars(\T,q)$. %
\end{enumerate}
\end{lemmarep}
\begin{proof}
First suppose that $f$ is "relevant" to $(\T,q)$ w.r.t.\ $\A$.  Let $S_q$ be a minimal support for $(\T,q)$ in $\A$
that contains $f$, witnessed by $h: q \homto \IsubA_{S_q,\T}$. Necessarily, every fact in $S_q$ must be 
"potentially relevant" to some atom in $q$ w.r.t.\ $S$, otherwise it could be removed, contradicting minimality of $S_q$. 
Now let $S^* \subseteq S_q$ be the set of facts in $S_q$ that are "potentially relevant" to at least one atom in $\intatoms(Q)$.
By definition, we have $f \in S^*$. 
If $|S^*| > k $, then we have more facts than "atoms" in $\intatoms(Q)$, so we could remove at least one fact from $S^*$ and still 
satisfy all "atoms" in $\intatoms(Q)$. Note here that we are exploiting the property of core DL-Lite dialects (like $\dlliter$)
where the minimal support of a query cannot contain more facts than the atoms in the query.
Moreover, by Lemma \ref{pot-relevant-cases}, the facts in $S^*$ cannot be used to satisfy the other "atoms" not in $\intatoms(Q)$. 
Thus, assuming $|S^*| > k $ contradicts our assumption of minimality, which yields $|S^*| \leq k$, as required. 

Next we note that from $h: q \homto \IsubA_{S_q,\T}$, we obtain a homomorphism 
$h^*: \intatoms(Q) \homto \IsubA_{S^*,\T}$, simply by setting $h^*(x) =h(x)$ for all variables $x$ in $\intatoms(Q)$. 
Indeed, as already noted, only the facts in $S^*$ are useful for satisfying the atoms in $ \intatoms(Q)$,
meaning that the image of $\intatoms(Q)$ under $h$ maps within $\IsubA_{S^*,\T}$.
It remains to prove that $S^*$ and $h^*$ verify the three conditions. 

For condition \ref{lem:potrel123a}, consider some $x \in \frontiervars(\T,q)$. 
Since $h: q \homto \IsubA_{S_q,\T}$, Lemma \ref{lem:sharedinabox} is applicable and yields $h^*(x)=h(x) \in \const(\A)$. 

For condition \ref{lem:potrel123b}, it suffices to consider the homomorphism $h: q \homto \IsubA_{S_q,\T}$. Indeed, 
by definition of $h^*$, we have $h^*(x) =h(x)$ for all variables $x$ in $\intatoms(Q)$, and in particular, for 
$x \in \frontiervars(\T,q)$. Thus, $h: q_{h^*} \homto \IsubA_{S_q,\T}$, which shows that $S_q, \T \models q_{h^*}$ (hence $\A, \T \models q_{h^*}$). 

Finally, to show condition \ref{lem:potrel123c}, suppose for a contradiction that there exist
$S' \subsetneq S^*$ and $h': \intatoms(Q) \homto \IsubA_{S',\T}$
such that $h'(x)=h^*(x)$ for every $x \in \frontiervars(\T,q)$. In this case, 
let us consider the set of facts $S_q^{-}= S' \cup (S_q \setminus S^*)$ and 
define $h''$ by setting $h''(x)=h'(x)$ for variables occurring in $\intatoms(Q)$,
and $h''(x)=h(x)$ for all remaining $x \in \vars(q)$. By definition, 
$h''$ agrees with $h$ on $ \frontiervars(\T,q)$. Thus, from
$h: q_{h^*} \homto \IsubA_{S_q,\T}$ and $h': \intatoms(Q) \homto \IsubA_{S',\T}$, 
we can infer that $h''$ is a homomorphism $q \homto \IsubA_{S_q^{-},\T}$.
This would mean that $S_q^{-}, \T \models q$, contradicting the minimality of $S_q$. 
We thus conclude that there can be no such $S' \subsetneq S^*$ and $h'$, so condition \ref{lem:potrel123b} is satisfied.

\medskip

For the other direction, let $S \subseteq \A$ and $h_S: \intatoms(Q) \homto \IsubA_{S,\T}$  satisfy all requirements of the lemma statement. 
In particular, $f \in S$ and $|S| \leq k$. We aim to show that $S$ can be extended to a minimal support of $(\T,q)$. 
To this end, let $S_{q}^{-}$ be a minimal support for $(\T, q_{h_S})$ in $\A$,
and $h_{q}^{-}: q_{h_S} \homto \IsubA_{S_{q}^{-},\T}$ a witnessing homomorphism (such $S_{q}^{-}$ and $h_{q}^{-}$ must exist due to condition \ref{lem:potrel123b}). 
Letting $S^* = S \cup  S_{q}^{-}$, we define a homomorphism $h: q \homto \IsubA_{S^*,\T}$
by setting $h(x) = h_S(x)$ if $x$ appears in $\intatoms(Q)$ and $h(x)= h_{q}^{-}(x)$ otherwise. 
Indeed, we know that $q_{h_S}$ was obtained by replacing each $x \in \frontiervars(\T,q)$ by $h_S(x)$,
so the homomorphisms $h_S$ and $h_{q}^{-}$ are compatible. It follows that $S^*$ is a support of 
$(\T,q)$, and it remains to show that $S^*$ is minimal. 

Let us thus suppose for a contradiction that there exists a fact $g \in S^*$ such that $S^* \setminus \{g\}$ is still a support for $(\T,q)$. 
If $g \in S^* \cap S_{q}^{-}$, then this implies that $S_{q}^{-}$ is not a minimal support. 
Indeed, we know that the facts in $S$ are not "potentially relevant" for any atoms in $q_{h_S}$,
so $S_{q}^{-} \setminus \{g\}$ must be a support for $(\T, q_{h_S})$, a contradiction. 
Thus, it must be the case that $g \in S$. Let $S' = S \setminus \{g\}$ and take any 
homomorphism $h_{S'}: \intatoms(Q) \homto \IsubA_{S',\T}$. Due to condition \ref{lem:potrel123b}, there must 
exist $x \in \frontiervars(\T,q)$ such that $h_{S'}(x) \neq h_{S}(x)$. Take $\beta \in \atoms(q) \setminus \intatoms(\T,q)$
with $x \in \vars(\beta)$. Since $\beta \in \intatoms(\T,q)$ and $h_S(x) \in \const(\A)$, 
any homomorphism of $q$ to $\IsubA_{S^* \setminus \{g\},\T}$ must send $x$ to $h_S(x)$, again yielding a contradiction. 
We can thus conclude that $S^*$ is a minimal support for $(\T,q)$, and since it contains $f$, 
this proves that $f$ is "relevant" $(\T,q)$ w.r.t.\  $S^*$ (and hence $\A$).
\end{proof}

This suggests the following procedure for deciding "relevance" of facts of type \ref{pot-relevant-casesb}: %
iterate (in logspace) over all candidate sets $S$ and homomorphisms $h_S$ and check whether the required conditions hold. 
We terminate either when some $S$ has been shown to satisfy the conditions (outputting `"relevant"') or when there are no more candidates to test 
(`not "relevant"'). Note that the second condition requires a "query evaluation" check, the cost of which will depend on the form of the 
query $q_{h_S}$. For the third condition, we must perform a second (logspace) exploration of possible 
$S' \subsetneq S$ and $h_{S'}$. 

If we consider "CQs" of bounded "interaction width", without further structural restrictions, 
we can obtain an optimal $\NP$ upper bound by implementing the sketched procedure using a
non-deterministic polytime Turing machine (using a non-deterministic guess to verify $(\A,\T) \models q_{h_S}$ in condition \ref{lem:potrel123b}).

For the $\logCFL$ and $\NL$ results, we further need to establish the complexity of 
"query evaluation" for structurally restricted classes of "OMQs" of bounded "interaction width":

\begin{theoremrep}\label{queryeval-biw}
"Query evaluation" is in $\logCFL$ for the class of "OMQs" $(\T,q) \in (\dlliter, \CQ)$ with "interaction width" at most $k$ and treewidth at most $m$ (for any fixed $k\geq 0$, $m \geq 1$). 
It is in $\NL$ if we further restrict to $(\T,q)$ such that $q$ is "acyclic" and has at most $\ell$ leaves (for fixed $\ell \geq 2$). 
\end{theoremrep}
\begin{proofsketch}
Assuming w.l.o.g.\ that $q$ is connected, 
we separately test for the existence of a homomorphism 
$h: q \homto \IsubA_{\A,\T}$ that fully maps $q$ into the anonymous part,
or one which sends at least one variable to an ABox constant. 
In the former case, we argue that $|q| \leq k$, 
so we can guess and check a (compact representation of a) 
potential homomorphism in $\NL$. For the second kind of homomorphism, 
we show that all variables are mapped either to an ABox constant or to an 
anonymous element $a P_{1} \ldots P_{n}$ with $n \leq k$, enabling the 
reuse of techniques developed for bounded-depth ontologies \cite{DBLP:journals/jacm/BienvenuKKPZ18}. 
\end{proofsketch}

\begin{proof}
Fix $k\geq 0$, $m \geq 1$, and $\ell \geq 2$. Note that it suffices to provide a procedure for "OMQs" whose component "CQ" is connected, since an arbitrary "OMQ" can be evaluated by separately checking whether each of its connected components is entailed. 
Let us thus take some $Q=(\T,q) \in (\dlliter, \CQ)$
with "interaction width" at most $k$,  %
and suppose w.l.o.g.\ that $q$ is a connected "CQ". 

We will devise procedures of the required complexity for
testing whether $\A \models Q$, or equivalently, $\A, \T \models q$.
By \Cref{canmod-prop}, the latter holds iff there is a homomorphism of $q$ into $\IsubA_{\A,\T}$. 
We shall distinguish two kinds of homomorphisms: those which map at least one variable to an "ABox" constant,
and those which map the whole query within the anonymous part. 

First, we consider how to check for the existence of a homomorphism
which maps at least one variable to an "ABox" constant. In this case, we can
observe that any such homomorphism cannot `reach' an anonymous element $a P_1 \ldots P_n$ if $n > k$,
because connectedness of the query means that we would also need to map variables onto $a$, $a P_1$, \ldots, 
$a P_1 \ldots P_{n-1}$. This would in turn mean that there would be $n$ role atoms that can all mapped within the anonymous part
of a single fact, implying that there are at least $n> k$ "atoms" in $\intatoms(\T,q)$. Since we only need to consider
homomorphisms that map to ABox constants and anonymous elements $a P_1 \ldots P_n$ with $n \leq k$,
we may reuse an existing procedures for $\dlliter$ designed for so-called bounded depth ontologies \cite{DBLP:journals/jacm/BienvenuKKPZ18},
where the maximum length of an anonymous element $a P_1 \ldots P_n$ cannot exceed a particular bound. 
In that paper, $\logCFL$ and $\NL$ procedures are devised respectively for the classes of OMQs with bounded treewidth CQs and acyclic CQs with bounded numbers of leaves, respectively, under the assumption that there is a constant bound on the ontology depth. 
Intuitively, these procedures 
consider mappings of the query variables into the ABox constants and anonymous elements. 
Crucially, the depth bound ensures that every anonymous element can be represented with logarithmic space,
Moreover, it is shown that when restricted to anonymous elements of constant length, 
it can be tested in $\NL$ whether a potential domain element $a P_1 \ldots P_n$ is actually present in 
$\Delta^{\IsubA_{\A,\T}}$ and whether a given such element (resp.\ pair of elements) belongs to $A^{\IsubA_{\A,\T}}$ (resp.\ $R^{\IsubA_{\A,\T}}$). 
Applying these procedures, we are able to determine the existence of a homomorphism that 
maps at least one variable to an "ABox" constant (or more generally, maps all query variables to 
anonymous elements $a P_1 \ldots P_n$ with $n \leq k$), in $\logCFL$ if $q$ has treewidth at most $m$ 
and in $\NL$ if $q$ is acyclic with at most $\ell$ leaves (recall that $m$ and $\ell$ are fixed constants). 

To complete the proof, we must show how to decide the existence of a homomorphism that 
maps the whole query within the anonymous part, possibly using anonymous elements whose length is greater than $k$. 
We first observe that this means that there is 
a single fact $f_q \in \A$ such that $q \homto \IsubA_{\{f_q\},\T}$. This would imply in turn 
that all "atoms" in $q$ belong to $\intatoms(\T,q)$,
and hence $|q| \leq k$. Thus, if $|q| > k$, we can immediately conclude that no such homomorphism 
exists. Otherwise, if $|q| \leq k$, then $q$ has a constant number of "atoms". 
Now if there is a homomorphism of $q$ which maps all variables into the anonymous part, 
then exists some role $T$ such that (a) there is an element 
$w=c P_1 \ldots P_n \in \Delta^{\IsubA_{\A,\T}}$ with $P_n=T$
and (b) it is possible to map $q$ into the subtree rooted at $w$ with some variable $v \in \vars(q)$
sent to $w$. Importantly, we may assume w.l.o.g.\ in (a) 
that $n \leq 2 |\T|$ (due to the way the words in $ \Delta^{\IsubA_{\A,\T}}$ are defined),
and in (b) we don't actually care about $c$ nor the $P_1 \ldots P_{n-1}$ but only about what holds starting from $P_N=T$.
This is because we are considering the subtree rooted of $\IsubA_{\A,\T}$ rooted at $w$, which is fully determined by the role~$T$.
But this means in turn that once we have decided that we map into a tree rooted at $w$ (or any other anonymous element whose final role is $T$),
then we only need to consider homomorphisms that map at most $|q|\leq k$ steps deeper into the anonymous part. 
And since we have only constantly many query atoms (hence variables) to consider,
we can actually guess the function $h$ representing a potential homomorphism, 
by guessing a possibly empty word of length $\leq k$ over $\NRpm$ for each variable in $q$,
and then checking whether $h$ is such that:
\begin{enumerate}
\item it only maps variables to words $R_1 \ldots R_l$ such that  $\T \models \exists T^- \sqsubseteq \exists R_1$,
and $\T \models \exists R_i^- \subseteq \exists R_{i+1}$ for $1 \leq i < l$. 
\item it satisfies the query atoms, meaning that if $A(x) \in q$ and $h(x)=R_1 \ldots R_l$, then $\T \models R_l^- \sqsubseteq A$ (or $\T \models \exists T^- \sqsubseteq A$ if $h(x)=\epsilon$),
and if $S(x,y) \in q$, then either: 
\begin{itemize}
\item $h(x)=R_1 \ldots R_l$ and $h(y) = R_1 \ldots R_{l+1}$ and $\T \models R_{l+1} \sqsubseteq S$, or
\item $h(x)=R_1 \ldots R_{l+1}$ and $h(y) = R_1 \ldots R_{l}$ and $\T \models R_{l+1}^- \sqsubseteq S$ 
\end{itemize}
\end{enumerate}
Thus, to sum up, we guess a role $T$ and verify that 
there exists some element $c P_1 \ldots P_n \in \Delta^{\IsubA_{\A,\T}}$ with $P_n=T$. 
This can be done in $\NL$ by first guessing $c$ and some $P_1$ such that  $(\A,\T) \models \exists P_1(a)$,
 and then sequentially guessing each $P_i$, keeping only two roles in memory at each time, checking whether 
$\T \models P_{i} \sqsubseteq P_{i+1}$. 
Then we guess a function $h$ of the form described above and check whether it satisfies the required conditions.
The function can be represented in logspace because $|q| \leq k$. Moreover, the conditions can also be checked in 
$\NL$ (using the fact that axiom entailment is in $\NL$). 
If we find a role $T$ and function $h$ satisfying the conditions, then this means that 
$q$ can be mapped into the anonymous part of the canonical model (we can obtain the witnessing homomorphism 
by taking any $w= c P_1 \ldots P_n \in \Delta^{\IsubA_{\A,\T}}$ with $P_n=T$ and then appending the word given by $h$).
Conversely, if such a homomorphism exists, then we can use it pick $T$ and $h$. 
We thus have an $\NL$ procedure for checking homomorphisms of this type. 

By sequentially executing the procedures for the two types of homomorphism, 
we obtain a procedure that can decide the existence of any homomorphism of $q$ into $\Delta^{\IsubA_{\A,\T}}$.
The resulting procedure runs in $\logCFL$ provided $q$ has treewidth $\leq m$,
and in $\NL$ if further $q$ is acyclic and with $\leq \ell$ leaves. 
\end{proof}

\begin{toappendix}
\subsubsection*{Further details on the proof of \Cref{omq-thm}}
To complete the proof, we need to explain how we can 
test whether $(\A,\T) \models q_{h_S}$, for the queries $q_{h_S}$ constructed in Lemma \ref{lem:potrel123}. 
Recall first that $q_{h_S}$ is obtained from $q$ by (a) removing all atoms in $\intatoms(\T,q)$,
then (b) replacing $x \in \frontiervars(\T,q)$ by $h_S(x)$. Clearly, (a) cannot adversely affect the query structure. 
So let us now consider $q'_{h_S}$ where we retain (a), but instead of (b), we do the following: 
 add the query atom $A_c(x)$, for each $x \in \frontiervars(\T,q)$ such that $h(x)=c$, with $A_c$ a fresh concept name. 
 We also define an ABox $\A'$ by extending $\A$ with the facts $\{A_c(c) \mid h(x)=c, x \in \frontiervars(\T,q)\}$. 
 It is straightforward to show that
 $$(\A,\T) \models q_{h_S} \quad \text{iff} \quad (\A', \T) \models q'_{h_S}$$
Importantly, however, $q'_{h_S}$ has the same treewidth as $q$, and if $q$ is acyclic, then $q'_{h_S}$ 
has the same number of leaves as~$q$. Thus, we can apply the $\logCFL$ and $\NL$ procedures from 
\Cref{queryeval-biw} to decide whether $(\A', \T) \models q'_{h_S}$. Finally, we note that we can construct $\A'$ from $\A$
and $q'_{h_S}$ from $q_{h_S}$ by means of a logspace transducer, completing the upper bound argument. 
\end{toappendix}

Although \Cref{queryeval-biw} cannot be used directly to evaluate $q_{h_S}$
(as the instantiation of shared variables can make an acyclic query become cyclic,
or increase the treewidth), it is possible to  
simulate constants using fresh unary predicates, thereby retaining $q$'s good structural properties
and obtaining
the desired  $\logCFL$ (resp.\ $\NL$) upper bounds.

It is worth observing that for "interaction-free" OMQs,%
we can only have type \ref{pot-relevant-casesa}  potentially relevant facts. Thus, we can
decide relevance using AQ checks (to establish 
potential relevance), followed by evaluating the OMQ $(\T, q_{\text{-}\alpha_f})$.
This makes relevance for this class of OMQs easily implementable on top of any OMQA system.

\section{Minimal Homomorphisms on Graphs}\label{sec:graph}

\Cref{th:cqsighard} shows that the "relevance problem" for "CQs" is a step higher in the polynomial hierarchy than the corresponding "query evaluation problem", and that the $\sigmaptwo$-hardness already holds for binary signatures. 
We will now study the problem for simplest kind of databases or queries: \emph{graphs}, either directed or undirected, which we call \reintro{digraphs} and \reintro{graphs}, respectively. We show that the problem remains $\sigmaptwo$-complete even on this simple setting, which can be seen as a result of independent interest.

\begin{toappendix}
    What we shall call a \AP""graph"" from now on and without further notice, is a simple undirected graph without any loop or isolated vertex, and a \AP""digraph"" the directed counterpart. In particular, a "digraph" can be seen equivalently either a "database" or a (constant-free, Boolean) "CQ" over a signature containing a single binary relation, but we will rather focus on "(di)graphs@graphs" in the developments of this section.
\end{toappendix}

For "digraphs" $G,G'$, a \AP""$G$-homomorphic image"" on $G'$ is any subgraph $\hat G'$ of $G'$ such that $V(\hat G')= Im(h)$ and $E(\hat G') = \set{(h(v),h(v')): (v,v') \in E(G)}$ for some $h: G \homto G'$. We shall sometimes write $h(G)$ to denote $\hat G'$.
The definition for "graphs" is analogous. Such "$G$-homomorphic image" is said to be ""minimal@@ghim"" if it does not strictly contain any other "$G$-homomorphic image".
We can now define the \AP""Minimal Homomorphism Problem"" %
in its directed ($\DMH$) and undirected ($\UMH$) versions.
\AP
\decisionproblem{$\intro*\UMH$ (resp.\ $\intro*\DMH$) problem}%
                {A pair $G,G'$ of "graphs" (resp. "digraphs"), and an edge $e \in E(G')$.}%
                {Is $e$ in some "minimal@@ghim" "$G$-homomorphic image" on $G'$?}

The problem $\DMH$ is the natural equivalent to the "relevance problem" on a single binary relation:
\begin{lemmarep}\label{lem:poly-equiv-DMG-relevance}
    There are logspace many-one reductions between $\DMH$ and the "relevance problem" for the class of constant-free "CQs" over a single binary relation.
\end{lemmarep}
\begin{proof}
    Let $G = (V,E)$ and $G' = (V',E')$ be "digraphs", with $V=\set{v_1,\dotsc, v_n}$ and $V' = \set{u_1,\dotsc, u_{n'}}$. 
    Consider the constant-free "CQ" $q_G$ obtained from $G$ by replacing each $(v_i,v_j) \in E$ with an atom $R(x_i,x_j)$, where $\set{x_i}_{i \in [n]}$ are pairwise distinct variables.
    Consider also $D_{G'}$ as the "database" obtained from $G'$ by replacing each $(u_i,u_j) \in E'$ with a fact $R(c_i,c_j)$, where $\set{c_i}_{i \in [n']}$ are pairwise distinct constants.
    Observe that an edge $(u_i,u_j)$ is in a "minimal@@ghim" "$G$-homomorphic image" on $G'$ if, and only if, the fact $R(c_i,c_j)$ is "relevant" for the "CQ" $q_G$ on the "database" $D_G$.

    Likewise, from any constant-free "CQ" $q$ over a single binary relation $R$ and database $D$, we can produce "digraphs" $G_q$, $G_D$ by replacing facts or atoms $R(t,t')$ with edges $(v_t,v_{t'})$, and it follows that a given fact is "relevant" if{f} the corresponding edge is in a "minimal@@ghim" $G_q$"-homomorphic image" on $G_D$.
\end{proof}
As we show, it remains $\sigmaptwo$-complete.
\begin{toappendix}
\subsection{Directed Graphs}\label{ssec:dirgraph}
Directed graphs are naturally very close to databases and queries on the signature $\{R\}$, but we can actually go further, be emulating, using a single relation, any fixed signature that only contains relations of arity 2.

\begin{lemma}\label{lem:dbtodigraph}
Let $q$ be a constant-free Boolean "CQ" and $D$ a database, both on the signature $\Sigma\defeq\{R_1,\dots, R_n\}$ only containing arity 2 relations. Then there exist two "digraphs" $G_q, G_D = (V_q,E_q), (V_D,E_D)$ such that:
\begin{enumerate}
\item\label{lem:dbtodigrapha} there exists a bijective mapping $\eta$ from the "homomorphic images" of $q$ in $D$ to the $G_q$"-homomorphic images" on $G_D$ that preserves the inclusion relation;
\item\label{lem:dbtodigraphb} for every fact $f\in D$, there exists $e_f\in E_D$ \st\ for every "homomorphic image" $S$ of $q$ in $D$, $f\in S$ iff $e_f \in \eta(S)$;
\item\label{lem:dbtodigraphc} if $q$ is isomorphic to $D$, then $G_q=G_D$.
\end{enumerate}
\end{lemma}
\begin{proof}
We build graphs $G_q,G_D$ that precisely encode $q$ and $D$ resp. The construction will be identical for both in order to satisfy condition \ref{lem:dbtodigraphc}, so we shall focus on $D$ for now. The construction from $q$ is identical except that constants must be replaced with existential variables. For this, we start with a set of vertices that are identical to the constants of $D$, which we call \AP""anchor vertices"". Then we connect these vertices with the gadgets depicted in \Cref{fig:dbtodigraph}, designed to imitate the different relations in $\Sigma$. We denote this transformation by $\eta$.

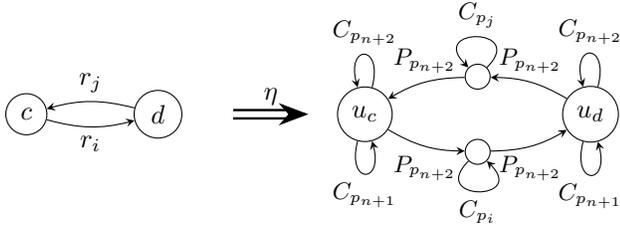
\begin{figure}[tb]
\centering
      \begin{tikzpicture}
\coordinate (00) at (0, 0);
\coordinate (01) at (-1, 0);
\coordinate (02) at (-3.75, 0);
\coordinate (03) at (-2, 0);
\coordinate (04) at (0.75, 0);
\coordinate (05) at (3.75, 0);
\coordinate (06) at (2.25, 0.5);
\coordinate (07) at (2.25, -0.5);
\begin{pgfonlayer}{nodelayer}
\node [] (0) at (00) {};
\node [] (1) at (01) {};
\node [draw, circle] (2) at (02) {$c$};
\node [draw, circle] (3) at (03) {$d$};
\node [draw, circle] (4) at (04) {$u_c$};
\node [draw, circle] (5) at (05) {$u_d$};
\node [draw, circle] (6) at (06) {};
\node [draw, circle] (7) at (07) {};
\end{pgfonlayer}
\begin{pgfonlayer}{edgelayer}
\draw [line width=.2ex,-{Stealth[length=2.5ex]}, double distance=.5ex, draw=black] (1.center) to node[midway, above] {$\eta$} (0.center);
\draw [->,>=stealth, bend right=15] (2) to node[midway, below] {$r_i$} (3);
\draw [->,>=stealth, bend right=15] (3) to node[midway, above] {$r_j$} (2);
\draw [->,>=stealth, bend right=15] (4) to node[midway, below] {\footnotesize $P_{p_{n+2}}$} (7);
\draw [->,>=stealth, in=-45, out=-135, loop] (7) to node[midway, below] {\footnotesize $C_{p_i}$} ();
\draw [->,>=stealth, bend right=15] (7) to node[midway, below] {\footnotesize $P_{p_{n+2}}$} (5);
\draw [->,>=stealth, bend right=15] (5) to node[midway, above] {\footnotesize $P_{p_{n+2}}$} (6);
\draw [->,>=stealth, in=135, out=45, loop] (6) to node[midway, above] {\footnotesize $C_{p_j}$} ();
\draw [->,>=stealth, bend right=15] (6) to node[midway, above] {\footnotesize $P_{p_{n+2}}$} (4);
\draw [->,>=stealth, in=105, out=75, loop] (4) to node[midway, above] {\footnotesize $C_{p_{n+2}}$} ();
\draw [->,>=stealth, in=-75, out=-105, loop] (4) to node[midway, below] {\footnotesize $C_{p_{n+1}}$} ();
\draw [->,>=stealth, in=105, out=75, loop] (5) to node[midway, above] {\footnotesize $C_{p_{n+2}}$} ();
\draw [->,>=stealth, in=-75, out=-105, loop] (5) to node[midway, below] {\footnotesize $C_{p_{n+1}}$} ();
\end{pgfonlayer}
\end{tikzpicture}%
   
\caption{Construction for \Cref{lem:dbtodigraph}. The loops labelled $C_p$ are simple cycles of length $n$, and the arrows labelled $P_{p_{n+2}}$ are simple paths of length 7. The two "anchor vertices" are $u_c$ and $u_d$.}
\label{fig:dbtodigraph}
\end{figure}

We can easily see that $\eta$ is a bijection between the subsets of $D$ and set $\S$ of subgraphs $S$ of $D_G$ such that:
\begin{enumerate}
\item for every "anchor vertex" in $S$, $S$ also contains its associated $C_{p_{n+1}}$ and $C_{p_{n+2}}$;
\item for every "anchor vertex" in $S$, $S$ also contains at least one $P_{p_{n+2}}$ that touches it;
\item for every $P_{p_{n+2}}$ in $S$, $S$ also contains the coupled $P_{p_{n+2}}$ and the intermediate $C_p$.
\end{enumerate}
The image of $\eta$ always has this structure indeed, and $\eta^{-1}(S)$ can be obtained from every $S\in \S$ by simply removing all $C_{p_{n+1}}$s and $C_{p_{n+2}}$s then replacing every $\xrightarrow{P_{p_{n+2}}}C_{p_i}\xrightarrow{P_{p_{n+2}}}$ by $\xrightarrow{R_i}$. We now need to show that every $G_q$"-homomorphic image" $\hat G'$ on $G_D$ is in $\S$, that $\eta^{-1}(\hat G')$ is a "homomorphic image" of $q$ in $D$, and vice versa.

\begin{claim}\label{clm:dbtodigraph}
   For any prime number $p\le p_{n+2}$, the only $C_p$"-homomorphic images" on $G_D$ are the simple $C_p$.
\end{claim}
\begin{nestedproof}
Let $p\le p_{n+2}$ be a prime number. Consider the "homomorphic image" $P$ of a path of length $p$ in $G_D$. Since the "anchor vertices" are separated by $p_{n+2}\ge p$ edges, $P$ can only contain one of them (or two if they are both its endpoints). Additionally, $P$ cannot contain both the $C_{p_{n+1}}$ and the $C_{p_{n+2}}$ attached to an "anchor vertex" because that would require $p_{n+1}+p_{n+2}>p$ edges. $P$ is therefore contained in a subgraph of $G_D$ of the form $\xrightarrow{P_{p_{n+2}}}C_{p'}\xrightarrow{P_{p_{n+2}}}$ for some $p'$. Now the only strongly connected components of such a subgraph is the $C_{p'}$ cycle; this means that any $C_{p}$"-homomorphic image" on $G_D$, which is a particular case of a $C_{p}$"-homomorphic image" since $P_p \homto C_p$, must be contained in $C_{p'}$. In fact it must be $C_{p'}$ as a whole since no proper subset is strongly connected, and, since $p$ and $p'$ are prime numbers, we have $C_{p'}\homto C_{p'}$ iff $p=p'$.
\end{nestedproof}

Let $h(G_q)$ be a $G_q$"-homomorphic image" on $G_D$. By \Cref{clm:dbtodigraph}, it must map every "anchor vertex" of $G_q$ to some "anchor vertex" of $G_D$ because they are the only constants at the intersection of a $C_{p_{n+1}}$ and a $C_{p_{n+2}}$. Then every $u_c \xrightarrow{P_{p_{n+2}}} C_p \xrightarrow{P_{p_{n+2}}} u_d$ must map to a $u_{h(c)} \xrightarrow{P_{p_{n+2}}} C_p \xrightarrow{P_{p_{n+2}}} u_{h(d)}$. At last, every "anchor vertex" in $G_q$ touches a $\xrightarrow{P_{p_{n+2}}}$ (otherwise there would be a variable of $q$ that appears in no atom) so the same can be said of their image by $h$. We have thus established that $h(G_q)\in \S$. 
Further, observe that the $i$-th prime number $p_i$ can be computed in polynomial time because it is bounded by $O(i \, \log(i))$ (eg. by \cite[Theorem 2]{rosser1939}) and testing primality is trivially in polynomial time since we consider $i$ which is smaller or equal to the number of binary relations, hence smaller than the input size.%

If we now consider the database $\eta^{-1}(h(G_q))$, it will have a fact $R_i(c,d)$ iff there is a $u_c \xrightarrow{P_{p_{n+2}}} C_{p_i} \xrightarrow{P_{p_{n+2}}} u_d$ in $h(G_q)$ iff there exist $x,y$ \st\ $h(x)=c$, $h(y)=d$ and $u_x \xrightarrow{P_{p_{n+2}}} C_{p_i} \xrightarrow{P_{p_{n+2}}} u_y$ in $G_q$ iff there exist $x,y$ \st\ $h(x)=c$, $h(y)=d$ and $R_i(x,y)\in q$. This means that $\eta^{-1}(h(G_q))$ is a "homomorphic image" of $q$ in $D$ indeed. Reciprocally, at last, $\eta(S)$ for any "homomorphic image" $S$ of $q$ in $D$ is a $G_q$"-homomorphic image" on $G_D$ by construction. At last, the inclusion relation is trivially preserved since $\eta(S)$ is always a direct encoding of $S$.

This concludes the proof of condition \ref{lem:dbtodigrapha}. For condition \ref{lem:dbtodigraphb}, we simply need to choose $e_f$ for $f=R_i(c,d)$ inside the corresponding $u_c \xrightarrow{P_{p_{n+2}}} C_{p_i} \xrightarrow{P_{p_{n+2}}} u_d$, because the whole structure will be present in $\eta(S)$ iff $f\in S$. As for condition \ref{lem:dbtodigraphc}, given that the construction on $q$ and $D$ is identical, we only need to make sure that we build $G_q$ and $G_D$ on the same set of "anchor vertices" from the beginning, in accordance with the isomorphism between $q$ and $D$.
\end{proof}

We obtain a new graph-theoretic $\Sigtwo$-complete problem:
\end{toappendix}
\begin{theoremrep}\label{th:dmhsigtwo}
$\DMH$ is $\Sigtwo$-complete.
\end{theoremrep}
\begin{proofsketch}
   The upper bound is straightforward. For the lower bound, we reduce from the "relevance problem" over binary relations $R_1,\dotsc,R_n$, known to be $\sigmaptwo$-hard (\Cref{th:cqsighard}). Given an input $D,q,\aFact$ of the "relevance problem" we consider "digraphs" $G_q,G_D$ and an edge $e_\aFact \in E(G_D)$ so that this is a positive instance of $\DMH$ if{f} $D,q,\aFact$ is positive for "relevance". The idea is simply to replace each atom or fact $R_i(a,b)$ with a long directed path from $a$ to $b$ containing a short directed cycle of a prime length $p_i$ in the middle, where $p_1 < \dotsb < p_n$ are the first $n$ prime numbers (which can be computed in polynomial time due to \cite[Thm. 2]{rosser1939}).\IfApx{\footnote{The $i$-th prime number $p_i$ can be computed in polynomial time since it is bounded by $O(i \, \log(i))$ (cf.\ \cite[Thm. 2]{rosser1939}) and testing primality is trivially in polynomial time since $i$ is smaller than the number of binary relations, and hence than the input size.} The long paths are of length $p_n$ to ensure that the only `short' cycles of length $\leq p_n$ are the ones of these gadgets.}{}
   In this way, we can build both $G_q$ from $q$ and $G_D$ from $D$. 
   Since prime $p_i$-cycles can only be homomorphically mapped to $p_i$-cycles, and the long paths prevent creating any other short cycle, a homomorphism $q \homto D$ can be equivalently seen as a homomorphism $G_q \homto G_D$ from prime cycles to prime cycles. By taking any edge $e_\aFact$ of the cycle contained in the replacement of fact $\aFact$ in $G_D$, we have that $e_\aFact$ belongs to a "minimal@@ghim" $G_q$"-homomorphic image" on $G_D$ if{f} $\aFact$ is relevant for $q$ in $D$.
\end{proofsketch}
\begin{proof}
We reduce from the "relevance problem" for constant-free Boolean "CQs" of arity 2 (\Cref{th:cqsighard}).
Given a constant-free Boolean "CQ" $q$ on the signature $\Sigma$, a database $D$ and a fact $f\in D$, we can first assume that $D$ is on the same signature because any other fact would be trivially "irrelevant" to $q$. Then we can apply \Cref{lem:dbtodigraph} to get our instance $G_q,G_d,e_f$ of $\DMH$. From conditions \ref{lem:dbtodigrapha} and \ref{lem:dbtodigraphb}, we get that $e_f$ is in some "minimal $G_q$-homomorphic image" on $G_D$ iff $f$ is "relevant" to $q$ in $D$.
\end{proof}
As a corollary, we have that the "relevance problem" is already hard for "CQs" using a single binary relation.
\begin{corollary}[of \Cref{th:dmhsigtwo,lem:poly-equiv-DMG-relevance}]
    The "relevance problem" for "CQs" is $\sigmaptwo$-complete even on signatures having a single binary relation.
\end{corollary}

\begin{toappendix}
\subsection{Undirected Graphs}
We show that the $\sigmaptwo$ hardness persists even for the "Minimal Homomorphism Problem" over undirected "graphs".    
\end{toappendix}

\begin{theoremrep}
    $\UMH$ is $\sigmaptwo$-complete.
\end{theoremrep}
\begin{proofsketch}
    The lower bound is more involved than for %
    $\DMH$ since for undirected "graphs" we no longer have that prime $p$-cycles can only be homomorphically mapped to $p$-cycles (in fact they can %
    map to any smaller odd cycle). We can however reduce from $\DMH$ by making use of graph-theoretic techniques known as `replacement methods'. 
\end{proofsketch}
\begin{proof}
    The upper bound is straightforward. 
    For the lower bound, we reduce from the directed version $\DMH$ of the problem already shown in \Cref{th:dmhsigtwo} to be $\sigmaptwo$-hard. 
    Let $G,G',(u',v')$ be an instance of $\DMH$. 
    We shall base the proof on some known graph-theoretic techniques collectively known as `replacement methods'. 
    The following preliminary definitions can be found in \cite[§4.4]{hell2004graphs}.
    Let us define a \AP""replacement graph"" as any triplet $J,s,t$ where $J$ is a "graph" and $s,t$ two distinct vertices thereof. For such a "replacement graph" and a "digraph" $H$, we denote by \AP$\intro*\starRep{H}{J}$ the (undirected) "graph" obtained from $H$ by replacing each directed edge $(x,y) \in E(H)$ by an isomorphic copy $J_{xy}$ of $J$, identifying $x$ with $s$ and $y$ with $t$; assuming that all the copies $J_{xy}$ are pairwise vertex disjoint (except perhaps for $x,y$). 
    Observe that in the notation $\starRep{H}{J}$ the vertices $s,t$ of $J$ are implicit and shall always be clear from the context.
    For any "digraph" "homomorphism" $h: H \homto H'$ let us define the "graph" "homomorphism" \AP$\intro*\homstarRep{h}{J} : \starRep{H}{J} \homto \starRep{H'}{J}$ as:
    \begin{itemize}
        \item $\homstarRep{h}{J}(w) = h(w)$ for every $w \in V(H)$,
        \item for any other vertex $w \in V(\starRep{H}{J}) \setminus V(H)$ which is in a $J_{xy}$ copy in $\starRep{H}{J}$, we define $\homstarRep{h}{J}(w)$ to be the corresponding vertex of the $J_{h(x)h(y)}$ copy in $\starRep{H'}{J}$.
    \end{itemize}

    \begin{claim}\label{cl:goodreplacement-undirected}
    There exists a "replacement graph" $\AP\intro*\goodJ,s,t$ such that every "graph" "homomorphism" $\starRep{G}{\goodJ} \homto \starRep{H}{\goodJ}$ is equal to $\homstarRep{f}{\goodJ}$ for some "digraph" "homomorphism" $f: G \homto G'$. Further, $\goodJ$ does not depend on $G,G'$ and contains an edge non-adjacent to $\set{s,t}$.
    \end{claim}
    \begin{nestedproof}
        \AP
        The statement follows directly from well-known results. Without getting into the details of the definitions, it suffices to know that a "replacement graph" may be ""strong@@graph"" or not, and that a "graph" may be ""rigid@@graph"" or not. There exist "replacement graphs" which are both "strong@@graph" and "rigid@@graph". For example, \cite[Figure 4.7]{hell2004graphs} defines a family of "replacement graphs" $G_k,a,b$ which are all both "strong@@graph" \cite[Proposition 4.16]{hell2004graphs} and "rigid@@graph" \cite[Proposition 4.6]{hell2004graphs} (and contain edges non-incident to $a,b$). The claim then follows from taking any such $G_k,a,b$ as $\reintro*{\goodJ},s,t$ and applying \cite[Proposition 4.12]{hell2004graphs}, whose statement reads:
            \textsl{For every "rigid@@graph" "strong@@graph" "replacement graph" $J,s,t$ and every pair $H,H'$ of "digraphs", we have that every "homomorphism" $\starRep{H}{J} \homto \starRep{H'}{J}$ is equal to $\homstarRep{f}{J}$ for some "homomorphism" $f:H \homto H'$.}
    \end{nestedproof}

    Consider the graphs $\starRep{G}{\goodJ}$, $\starRep{G'}{\goodJ}$ given by the claim above, and let $\set{\hat u',\hat v'}$ be an edge of the $\goodJ_{u'v'}$ copy of $\goodJ$ in $\starRep{G'}{\goodJ}$ non-adjacent to $u',v'$.
    \begin{claim}\label{cl:reduction-minhom-wrt-star}
        $G,G',(u',v')$ is a positive instance of the $\DMH$ problem if, and only if, $\starRep{G}{\goodJ},\starRep{G'}{\goodJ},\set{\hat u',\hat v'}$ is a positive instance of the $\UMH$ problem.
    \end{claim}
    \begin{nestedproof}
        \begin{figure}
      \begin{tikzpicture}
\coordinate (00) at (-3, 3);
\coordinate (01) at (-2, 3);
\coordinate (02) at (-1, 3);
\coordinate (03) at (1, 3);
\coordinate (04) at (2, 3);
\coordinate (05) at (3, 3);
\coordinate (06) at (2, 2);
\coordinate (07) at (-3.5, -1.75);
\coordinate (08) at (-1.5, -1.75);
\coordinate (09) at (0.5, -1.75);
\coordinate (010) at (0, 0.25);
\coordinate (011) at (2, 0.25);
\coordinate (012) at (4, 0.25);
\coordinate (013) at (2, -1.75);
\coordinate (014) at (-1, -1.75);
\coordinate (015) at (-0.25, -1.75);
\coordinate (016) at (2, -1);
\coordinate (017) at (2, -0.25);
\coordinate (018) at (-2, 3.75);
\coordinate (019) at (2.5, 3.75);
\coordinate (020) at (-1.5, -1);
\coordinate (021) at (3, 1);
\coordinate (022) at (3, 2);
\coordinate (023) at (4, -1.75);
\coordinate (024) at (0, 1.75);
\coordinate (025) at (0, 0.75);
\begin{pgfonlayer}{nodelayer}
\node [draw, circle] (0) at (00) {};
\node [label={[yshift=-.5ex]\footnotesize $u$}, draw, circle] (1) at (01) {};
\node [label={[yshift=-.5ex]\footnotesize $v$}, draw, circle] (2) at (02) {};
\node [draw, circle] (3) at (03) {};
\node [label={[yshift=-.5ex]\footnotesize $u'$}, draw, circle] (4) at (04) {};
\node [draw, circle] (5) at (05) {};
\node [label={[below, yshift=-1.5ex, xshift=.3ex]\footnotesize $v'$}, draw, circle] (6) at (06) {};
\node [draw, circle,fill=white] (7) at (07) {};
\node [label={[yshift=-.5ex]\footnotesize $u$}, draw, circle,fill=white] (8) at (08) {};
\node [label={[yshift=-.5ex]\footnotesize $v$}, draw, circle,fill=white] (9) at (09) {};
\node [draw, circle,fill=white] (10) at (010) {};
\node [label={[yshift=-.5ex]\footnotesize $u'$}, draw, circle,fill=white] (11) at (011) {};
\node [draw, circle,fill=white] (12) at (012) {};
\node [label={[below, yshift=-1.5ex, xshift=.3ex]\footnotesize $v'$}, draw, circle,fill=white] (13) at (013) {};
\node [] (18) at (018) {$G$};
\node [] (19) at (019) {$G'$};
\node [] (20) at (020) {$G*J$};
\node [] (21) at (021) {$G'*J$};
\node [draw, circle,fill=white] (22) at (022) {};
\node [draw, circle,fill=white] (23) at (023) {};
\end{pgfonlayer}

\newcommand{\Jsubgraph}[2]{%
   \draw [color=red,bend left,decorate,decoration={bumps, amplitude=.4ex, segment length=2ex}] (#1.center) to node[midway, below=-.1, sloped] {\tiny $J$} (#2.center);
   \draw [color=red,bend right,decorate,decoration={bumps, amplitude=-.4ex, segment length=2ex}] (#1.center) to (#2.center);
   \path [color=red] (#1) to node[pos=.05] {\tiny $s$} node[pos=.95] {\tiny $t$} (#2);
}
\begin{pgfonlayer}{edgelayer}
\Jsubgraph{8}{9}
\Jsubgraph{7}{8}
\Jsubgraph{10}{11}
\Jsubgraph{11}{13}
\Jsubgraph{11}{12}
\Jsubgraph{13}{23}
\Jsubgraph{23}{12}
\path (8) to coordinate[pos=.2,yshift=-.75ex] (0014) coordinate[pos=.5,yshift=-1.25ex] (0015) (9);
\path (11) to coordinate[pos=.5,xshift=-1.25ex] (0016) coordinate[pos=.2,xshift=-.75ex] (0017) (13);
\end{pgfonlayer}
\begin{pgfonlayer}{nodelayer}
\node [label={[yshift=-.6ex]\tiny $\hat u$}, draw, circle, inner sep=.5ex] (14) at (0014) {};
\node [label={[yshift=-.6ex]\tiny $\hat v$}, draw, circle, inner sep=.5ex] (15) at (0015) {}; 
\node [label={[right,yshift=-.3ex]\tiny $\hat v'$}, draw, circle, inner sep=.5ex] (16) at (0016) {}; 
\node [label={[right,yshift=-.6ex]\tiny $\hat u'$}, draw, circle, inner sep=.5ex] (17) at (0017) {}; 
\end{pgfonlayer}

\begin{pgfonlayer}{edgelayer}
\draw [->,>=stealth] (0) to (1);
\draw [->,>=stealth] (1) to (2);
\draw [->,>=stealth] (3) to (4);
\draw [->,>=stealth] (4) to (6);
\draw [->,>=stealth] (4) to (5);
\draw [->,>=stealth] (6) to (22);
\draw [->,>=stealth] (22) to (5);
\draw [color=red,->,>=stealth] (14) to (15);
\draw [color=red,->,>=stealth] (17) to (16);
\draw [color=blue, dashed,->,>=stealth, bend right] (0) to (3);
\draw [color=blue, dashed,->,>=stealth, bend left] (1) to node[midway, above] {$f$} (4);
\draw [color=blue, dashed,->,>=stealth, bend right=15] (2) to (6);
\draw [color=blue, dashed,->,>=stealth, bend left=15] (7) to node[midway, above] {$f^{*J}$} (10);
\draw [color=blue, dashed,->,>=stealth, bend left=15] (8) to (11);
\draw [color=blue, dashed,->,>=stealth, bend right=15] (9) to (13);
\draw [color=blue, dashed,->,>=stealth, bend left=15] (15) to (16);
\draw [color=blue, dashed,->,>=stealth, bend left=15] (14) to (17);
\draw [line width=.2ex,-{Stealth[length=2.5ex]}, double distance=.5ex, draw=black] (024) to node[midway, left] {$*J$} (025);
\end{pgfonlayer}
\end{tikzpicture}%
   
            \caption{Visual support for proof of \Cref{cl:reduction-minhom-wrt-star}}
            \label{fig:fstar-homs}
        \end{figure}
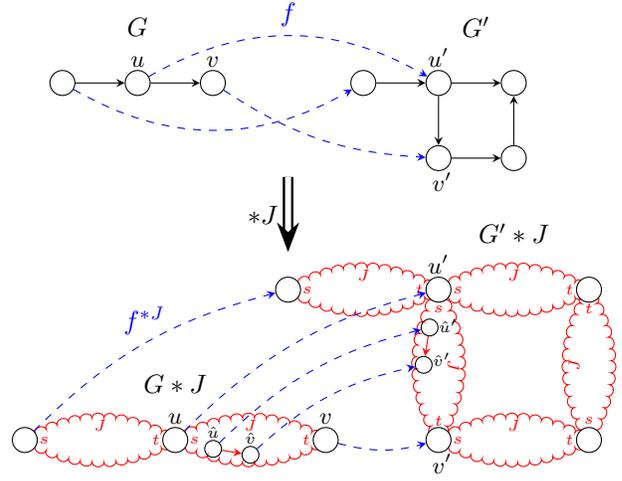
        From left to right, we have that $(u',v')$ belongs to a "minimal@@ghim" "$G$-homomorphic image" $h(G)$ via the "digraph" "homomorphism" $h: G \homto G'$. In particular, $h(u,v) = (u',v')$ for some $(u,v) \in E(G)$ -- cf.\ \Cref{fig:fstar-homs} for visual aid. The "graph" "homomorphism" $\homstarRep{h}{\goodJ} : \starRep{G}{\goodJ} \homto \starRep{G'}{\goodJ}$ must contain $\set{\hat u',\hat v'}$ by definition, since its image onto $\goodJ_{uv}$ is precisely the copy $\goodJ_{u'v'}$ where $\set{\hat u',\hat v'}$ lies. We now show that it is "minimal@@ghim". By means of contradiction, suppose it is not "minimal@@ghim", and hence that there is some "graph" "homomorphism" $g: \starRep{G}{\goodJ} \homto \starRep{G'}{\goodJ}$ such that $g(\starRep{G}{\goodJ}) \subsetneq \homstarRep{h}{\goodJ}(\starRep{G}{\goodJ})$. By \Cref{cl:goodreplacement-undirected} this means that $g = \homstarRep{f}{\goodJ}$ for some $f: G \homto G'$. By definition of $\homstarRep{f}{\goodJ}$ this would imply that $f(G) \subsetneq h(G)$ contradicting the "minimality@@ghim" of $h(G)$.

        From right to left, we have that $(\hat u',\hat v')$ belongs to a "minimal@@ghim" $\starRep{G}{\goodJ}$"-homomorphic image" $h(\starRep{G}{\goodJ})$ via the "graph" "homomorphism" $h: \starRep{G}{\goodJ} \homto \starRep{G'}{\goodJ}$. In particular, $h(\set{\hat u,\hat v}) = \set{\hat u',\hat v'}$ for some $\set{\hat u,\hat v} \in E(\starRep{G}{\goodJ})$ -- cf.\ \Cref{fig:fstar-homs} for visual aid.
        By \Cref{cl:goodreplacement-undirected}, $h = \homstarRep{f}{\goodJ}$ for some $f: G \homto G'$. 
        Let us show that $f(G)$ is a "minimal@@ghim" $G$"-homomorphic image" containing $(u',v')$.
        By definition of $\homstarRep{f}{\goodJ}$, $\set{\hat u,\hat v}$ is the edge of some copy $\goodJ_{uv}$ (for some $(u,v) \in E(G)$) corresponding to $\set{\hat u',\hat v'}$ in $\goodJ_{u'v'}$. Hence, $h(u,v) = (u',v')$ and thus $(u',v') \in f(G)$. If $f(G)$ is not "minimal@@ghim", meaning that there is some $g:G \homto G'$ with $g(G) \subsetneq f(G)$, note that $\homstarRep{g}{\goodJ}(\starRep{G}{\goodJ}) \subsetneq \homstarRep{f}{\goodJ}(\starRep{G}{\goodJ})$, which is in contradiction with our hypothesis.
    \end{nestedproof}
    Since, by \Cref{cl:goodreplacement-undirected}, $\goodJ$ is fixed and independent of $G,G'$, we have that \Cref{cl:reduction-minhom-wrt-star} yields a polynomial-time reduction and we conclude the proof.
\end{proof}

\section{Conclusion and Discussion}
\label{sec:discussion}
The main takeaway from our study is that the difficulty of the relevance problem for ($\dlliter$-mediated) CQs hinges fundamentally on the number of atoms (or of variables from different atoms) that may interact with a fact in the following sense: if we restrict the number of self-join variables or interacting atoms,
then %
the complexity of relevance is essentially as good or bad as that of query evaluation. These insights allowed us to identify natural classes of queries for which   relevance can be efficiently decided, and we expect that they will also prove useful when designing practical algorithms for the relevance problem. %
\medskip

\noindent\textbf{Complexity Dichotomies \& Parameterized Complexity} It is a natural question whether bounding the self-join width or interaction width is a necessary condition for obtaining tractability or lower complexity. More precisely, one might try to prove a dichotomy result along the following lines: ``For every recursively enumerable class C of conjunctive queries, if C has bounded self-join width and bounded (hyper)treewidth, then relevance is in polynomial time, and otherwise it is NP-hard''. Unfortunately, it is not at all clear that such a dichotomy holds, as  
there is currently no such %
dichtomy known for the combined complexity of CQ evaluation, %
and existing research provides strong evidence against its existence %
\cite{DBLP:journals/jacm/Grohe07}.  
By contrast, an FPT / W[1]-hard dichotomy has been proven for the \emph{parameterized} evaluation of CQs, %
using the treewidth (modulo cores) as the parameter~\cite{DBLP:journals/jacm/Grohe07}. Hence, a promising direction %
is to investigate the parameterized complexity of the relevance problem.\medskip

\noindent\textbf{Other Ontology Languages} %
We conjecture that the notion of interaction width we introduced for $\dlliter$ can be fruitfully applied also to linear existential rules, as they similarly enjoy the singleton support property (i.e.\ only one fact is needed to satisfy a query atom). For ontology languages which admit conjunction, and hence do not enjoy the singleton support property (since multiple facts may needed to infer a query atom), 
such as $\mathsf{DL\text{-}Lite_{Horn}}$, new restrictions will be needed to define well-behaved OMQ classes admitting tractable relevance computation. 
Indeed, as noted in Section \ref{sec:sources}, the NP-hardness of relevance for atomic queries precludes 
any positive results for classes of $\mathsf{DL\text{-}Lite_{Horn}}$ OMQs obtained via bounding parameters dominated by the query size.
For non-first-order-rewritable DLs like $\mathcal{EL}$, Proposition \ref{prop:elhard} effectively rules out any tractability results in combined complexity, but it may 
 still be worthwhile %
 to develop pragmatic algorithms for deciding relevance in such logics. 
\medskip

\noindent\textbf{Relevance w.r.t.\ Inconsistent KBs} When studying relevance of OMQs, we focused on \emph{consistent} knowledge bases. %
Since explaining the inconsistency of a given KB can be seen as explaining the Boolean query `is the KB consistent?', 
we could naturally consider the relevance problem for inconsistency. We should be able to transfer upper bounds from the query setting, 
but not necessarily lower bounds, %
as such `inconsistency queries' may be restricted in ways that lower the complexity.  
Indeed, in $\dlliter$, relevance of inconsistency is easily shown to be in $\NL$ as minimal inconsistent subsets are of size at most~2. 
Alternatively, one may study the relevance problem for explanations of query answers under inconsistency-tolerant semantics.
Interestingly, existing proposals of explanations for repair-based semantics \cite{DBLP:journals/jair/BienvenuBG19} are defined via minimal supports, so we expect that our results will prove useful.%
\medskip %

\noindent\textbf{Interaction Width Beyond Relevance} We expect that our new notion of interaction width may prove useful for identifying tractable OMQ classes for other explanation-related tasks. In particular, it was recently shown that the class of interaction-free OMQs admits efficient counting of minimal supports \cite{ourKR25}, and we are cautiously optimistic that this positive result can be extended classes of OMQs in $\dlliter$ with bounded interaction width (in line with an analogous result shown in the database setting for classes of CQs with bounded self-join width \cite{ourpods25paper}). As detailed in the cited works, counting minimal supports has applications in query answer explanation as the number of minimal supports can be used to assign numeric scores to %
facts based upon their contribution to making the query hold. %

\section*{Acknowledgements}
This work has been partially supported by ANR grants INTENDED (ANR-19-CHIA-0014) and EXPAND (ANR-25-CE23-1215).

\ifarxiv
\else
  \section*{AI Declaration}
  The authors have not employed any Generative AI tools.
\fi

\ifarxiv
  \bibliographystyle{kr}
\else
  \bibliographystyle{kr}
\fi
\bibliography{long,biblio}
\appendix
\end{document}